\documentclass[10pt]{article}
\usepackage{natbib}
\usepackage{graphicx}
\usepackage{float}
\usepackage{bm}
\usepackage{mathtools}
\newtagform{k}{(}{$k$)}

\usepackage[osf,sc]{mathpazo}

\usepackage[scaled=0.90]{helvet}
\usepackage{titlesec}
\usepackage[scaled=0.85]{beramono}

\newcommand{\R}{{\mathbb R}}
\usepackage[T1]{fontenc}
\usepackage{setspace}
\usepackage{multicol}
\usepackage{lastpage}
\usepackage{url}
\usepackage{hyperref}
\usepackage{lastpage}
\usepackage{amsmath}
\usepackage{layout}

\usepackage{amsthm}
\theoremstyle{definition}
\newtheorem{definition}{Definition}[section]
\newtheorem{thm}{Theorem}[section]
\newtheorem{lemma}{Lemma}[section]

\newcommand{\blind}{1}

\titleformat{\section}
   {\sc\Large\center}{\thesection.}{1em}{}

\usepackage{array, xcolor}
\usepackage{color,hyperref}
\hypersetup{colorlinks,breaklinks,linkcolor=red,urlcolor=red,anchorcolor=red,citecolor=black}

\begin{document}

\def\spacingset#1{\renewcommand{\baselinestretch}%
{#1}\small\normalsize} \spacingset{1}
\bibliographystyle{agsm}
\bibpunct{(}{)}{;}{a}{}{,}

\if1\blind
{
  \title{\bf Local biplots for multi-dimensional scaling, with application to the microbiome}
  \author{Julia Fukuyama\hspace{.2cm}\\
    Department of Statistics, Indiana University Bloomington}
  \maketitle
} \fi

\if0\blind
{
  \bigskip
  \bigskip
  \bigskip
  \begin{center}
    {\LARGE\bf Local biplots for multi-dimensional scaling, with application to the microbiome}
\end{center}
  \medskip
} \fi

\bigskip
\begin{abstract}
    We present local biplots, a an extension of the classic principal components biplot to multi-dimensional scaling.
  Noticing that principal components biplots have an interpretation as the Jacobian of a map from data space to the principal subspace, we define local biplots as the Jacobian of the analogous map for multi-dimensional scaling.  
  In the process, we show a close relationship between our local biplot axes, generalized Euclidean distances, and generalized principal components.
  In simulations and real data we show how local biplots can shed light on what variables or combinations of variables are important for the low-dimensional embedding provided by multi-dimensional scaling.
  They give particular insight into a class of phylogenetically-informed distances commonly used in the analysis of microbiome data, showing that different variants of these distances can be interpreted as implicitly smoothing the data along the phylogenetic tree and that the extent of this smoothing is variable.
\end{abstract}

\noindent%
\vfill

\newpage
\spacingset{1.45} 

\section{Introduction}
\label{sec:intro}




Exploratory analysis of a high-dimensional dataset is often performed by defining a distance between the samples and using some form of multi-dimensional scaling (MDS) to obtain a low-dimensional representation of those samples.
The resulting representation of the data can then be used for visualization and hypothesis generation, and the distances themselves can be used for hypothesis testing or other downstream analyses.
However, all forms of multi-dimensional scaling have the limitation that they only provide information about the relationships between the samples.
The analyst would often like to know about the relationships between samples and variables as well: if we see gradients or clusters in the samples, what variables are they associated with?
Do particular regions in the map correspond to particularly high or low values of certain variables?
If principal components analysis (PCA) or a related method is used, these questions are answered naturally with the PCA biplot \citep{Gabriel1971-bt}, but multi-dimensional scaling is silent on these questions.

Related to the question of interpreting a single multi-dimensional scaling plot is the question of interpreting differences between multi-dimensional scaling plots that use different distances to represent the relationships among the same set of samples.
This issue can arise in any problem for which multiple distances are available for the same task, but the work here is motivated by analysis of microbiome data.
In the microbiome field, many distances are available, and multiple distances are often used on the same dataset.
For instance, it is common to pair a distance that uses only presence/absence information with one that uses abundance information and to attribute differences in the representations to the influence of rare species (as in for example \citet{lozupone2007quantitative}).
However, it is often unclear whether such conclusions are warranted, and an analog of the principal components biplot to multi-dimensional scaling would give us much more insight into the differences between the representations of the data given by multi-dimensional scaling with different distances.


Other authors have done work on this problem, and the local biplots defined here are most directly related to the non-linear biplots described in \citet{Gower1988-fj}.
These biplots are based on the idea that for any multi-dimensional scaling plot, we can define a function that takes a new data point and adds it to the multi-dimensional scaling embedding space.
Other suggestions for biplots for multi-dimensional scaling have tended to generalize the SVD interpretation of PCA, including \citet{Satten2017-ds} and \citet{wang-2019-gmd-biplot}.
Another simple and commonly-used method for visualizing the relationship between the variables and the embedding space is simply to compute correlations between variables and the sample scores on the embedding axes \citep{mccune2002analysis,daunis2011including}.
All of these methods have their merits, but we believe that the local biplots defined here represent a natural and as yet unexploited way of visualizing the relationship between the data space and the embedding space.


In this paper, we present an interpretation of the PCA biplot axes as the Jacobian of a map from data space to embedding space.
We then generalize this interpretation of the biplot axes, showing how we can obtain analogous axes in classical multi-dimensional scaling.
The resulting visualization tool has a natural interpretation in terms of sensitivities of the embeddings to values of the input variables, gives the analyst a picture of the relationships between the variables and samples and to different parts of the plot, and suggests further diagnostic tools for investigating the importance and linearity of the relationship between the data space and the embedding space defined by classical scaling.
Along the way, we generalize the classic result about the relationship between MDS with the Euclidean distance and principal components to MDS with a generalized Euclidean distance and generalized principal components.

\section{Notation}

We will use the following notation:
\begin{itemize}
\item Upper-case bold symbols (e.g. $\mathbf A$) represent matrices.
\item Lower-case bold symbols (e.g. $\mathbf x$) represent vectors.
\item Lower-case italic symbols represent the elements of a matrix or a vector ($a_{ij}$ is the scalar in the $i$th row, $j$th column of the matrix $\mathbf A$, $x_i$ is the $i$th element of the vector $\mathbf x$).
\item $\mathbf e_j$ is a vector with a 1 in position $j$ and 0's in all other positions.
\item $\mathbf 1_n$ is an $n$-vector filled with 1's.
\item $\mathbf 0_n$ is an $n$-vector filled with 0's.
\item $\mathbf I_n \in \R^{n \times n}$ is the identity matrix.
\item $\mathbf A \succ 0$ indicates that $\mathbf A$ is symmetric positive definite.
\item If $\mathbf A \in \R^{n \times p}$, $\mathbf A_{i \cdot} = \begin{pmatrix} a_{i1} & \cdots & a_{ip}\end{pmatrix}$, and $\mathbf A_{\cdot j} = \begin{pmatrix} a_{1j} \\ \vdots \\ a_{nj}\end{pmatrix}$.
\item $\mathbf A_{\cdot, k:l}$ indicates columns $k$ through $l$ of the matrix $\mathbf A$.
\item $\mathbf A_{k:l, m:n}$ indicates the submatrix of $\mathbf A$ consisting of the $k$ through $l$th rows and the $m$th through $n$th columns.
\item $\mathbf X$ will always denote a data matrix, with samples as rows and variables as columns.
  We will use $\mathbf x_i$ to denote the column vector $\mathbf X_{i\cdot}^T$, that is, the vector corresponding to sample $i$.
\item $\mathbf C_n = \mathbf I_n - \mathbf 1_n \mathbf 1_n^T / n$ is the centering matrix.
\item $d : \R^p \times \R^p \to \R^+$ will be used to denote an arbitrary distance function.
\item $d_{\mathbf y} : \R^p \to \R^+$, where $\mathbf y \in \R^p$ will denote the restriction of $d$ to $\{\mathbf y \} \times \R^p$.
\item $d_{\mathbf A} : \R^p \times \R^p \to \R^+$ where $0 \prec \mathbf A \in \R^{p \times p}$ will be a {\em generalized Euclidean distance}, with $d_{\mathbf A} : \mathbf x, \mathbf y \mapsto \sqrt{(\mathbf x - \mathbf y)^T \mathbf A (\mathbf x - \mathbf y)}$.
  Note that this is more commonly refered to as the Mahalanobis distance \citep{mahalanobis1936generalized}.
  We use ``generalized Euclidean distance'' instead because we do not wish to emphasize the probabilistic interpretation.
\item If $f : \R^m \to \R^n$ and $\mathbf z \in \R^n$, the Jacobian matrix $J_f(\mathbf z) \in \R^{n \times m}$ has elements $(J_f(\mathbf z))_{ij} = \frac{\partial f_i}{\partial z_j}(\mathbf z)$.
\end{itemize}

\section{Methods}
\label{sec:meth}

\subsection{Principal components biplots}

Before introducing local biplots for multi-dimensional scaling, we review principal components biplots.
There are several different types of PCA biplots, and we focus here on the form biplot.
Suppose we have a data matrix $\mathbf X \in \R^{n \times p}$ with centered columns.
Let $\mathbf X = \mathbf U \mathbf D \mathbf V^T$ be the singular value decomposition of $\mathbf X$, so that $\mathbf U \in \R^{n \times \text{rank}(\mathbf X)}$, $\mathbf V \in \R^{n \times \text{rank}(\mathbf X)}$, $\mathbf D \in \R^{\text{rank}(\mathbf X)\times \text{rank}(\mathbf X)}$, $\mathbf U^T \mathbf U = \mathbf V^T \mathbf V = \mathbf I_{\text{rank(X)}}$, and $\mathbf D$ is diagonal with $d_{ii} \ge d_{jj}$ for $i < j$.

In the form biplot, we display the rows of $\mathbf U_{\cdot, 1:k} \mathbf D_{1:k,1:k}$ and the rows of $\mathbf V_{\cdot, 1:k}$, usually with $k = 2$ in Cartesian coordinates.
In that case, the form biplot has the following properties, which are used for interpreting the relationships among the samples, among the variables, and between the variables and the samples:
\begin{itemize}
\item The position of the $i$th row of $\mathbf U_{\cdot, 1:2} \mathbf D_{1:2,1:2}$ along the $x$- and $y$-axes gives the projection of the $i$th row of $\mathbf X$ onto the first and second principal axes, respectively.
\item The position of the $j$th row of $\mathbf V_{\cdot,1:2}$ along the $x$- and $y$-axes gives the loading of the $j$th variable onto the first and second principal axes, respectively.
\item If $\hat{ \mathbf X}$ denotes the least-squares optimal rank-2 approximation of $\mathbf X$, then $\hat x_{ij} = \mathbf U_{\cdot, 1:2} \mathbf D_{1:2,1:2} \mathbf V_{\cdot,1:2}^T$ ($\hat x_{ij}$ is the inner product between the $i$th row of $\mathbf U_{\cdot, 1:2} \mathbf D_{1:2,1:2}$ and the $j$th row of $\mathbf V_{\cdot,1:2}$).
\end{itemize}
The last point in the list above provides the standard motivation behind biplots: if the matrix $\mathbf X$ is well approximated by a rank-2 matrix, then we can read off good approximations of elements of $\mathbf X$ from the biplot.
Even if $\mathbf X$ is not well approximated by a rank-2 matrix, the biplot still allows us to read off exactly (to the extent that we can compute inner products by looking at a plot) its rank-2 approximation.

To this list we add one more property.
Let $g : \R^p \to \R^k$ be the projection of a new data point onto the $k$-dimensional principal subspace, defined as $\mathbf z \mapsto (\mathbf V_{1:k})^T \mathbf z$.
Let $J_g$ be the Jacobian of that map.
$J_g(\mathbf z)$ describes the sensitivity of the projection of a new point $\mathbf z$ on the principal subspace to perturbations in the variables.
Then:
\begin{itemize}
\item For principal components, $J_g(\mathbf z) = J_g = (\mathbf V_{\cdot, 1:k})^T$ (the map is linear and so the Jacobian is constant).
  The $j$th column of $J_g$ gives the $j$th biplot axis in the principal components biplot.
\end{itemize}
This property is close to simply being the sensitivity of the embedding position of a row of $\mathbf X$ to a perturbation in the $j$th variable.
The difference is that we fix the principal axes when perturbing one of the variables.
In a more standard sensitivity formulation, perturbing one of the variables associated with one of the samples would change the principal subspace.
By thinking of the principal subspace as being fixed and defining a map that takes new data points into that space, we can describe what sorts of perturbations of supplemental points in the embedding space would result from perturbations of those points in the input data space.

It is this last property that we propose to generalize to multi-dimensional scaling.
Since multi-dimensional scaling is not a linear map from the data space to the embedding space, most of the biplot interpretations are not available, or they are only available in approximate forms.
The linearity of principal components corresponds to one set of sensitivities to perturbations of the input variables, no matter what the sample being perturbed is.
Since multi-dimensional scaling is non-linear, we will have different sensitivities to perturbations of the input variables in different parts of the data space, which will correspond to different sets of biplot axes in different parts of the space.
This complicates the representation, but it is unavoidable if we do not want to approximate a nonlinear technique by a linear one.

\subsection{Multi-dimensional scaling defines a map from data space to embedding space}

To generalize the Jacobian interpretation of the PCA biplot to classical multi-dimensional scaling, we need to define a function associated with a multi-dimensional scaling solution that maps from data space to embedding space.
We will assume that we have a matrix $\mathbf X \in \R^{n \times p}$, whose $i$th row is denoted $\mathbf x_i$ and a distance function $d: \R^p \times \R^p \to \R^+$.
Recall that in classical multi-dimensional scaling, as described in \citet{Gower1966-hx} and \citet{Torgerson1958-qf}, we start by defining $\mathbf \Delta \in \R^{n\times n}$ as
\begin{align}
 \delta_{ij} = d(\mathbf x_i, \mathbf x_j)^2 \label{Eq:squared-dist}
\end{align}
We then let
\begin{align}
-\frac{1}{2} \mathbf C_n \mathbf \Delta \mathbf C_n^T = \mathbf B \mathbf \Lambda \mathbf B^T \label{Eq:JDJ}
\end{align}
be the eigendecomposition of $-\frac{1}{2} \mathbf C_n \mathbf \Delta \mathbf C_n^T$, i.e., $\mathbf B \in \R^{n \times n}$ is an orthogonal matrix, $\mathbf \Lambda \in \R^{n \times n}$ is diagonal such that $\lambda_{ii} \ge \lambda_{jj}$ if $i > j$.
If the distances are Euclidean embeddable, all the diagonal elements of $\mathbf \Lambda$ will be non-negative, and the Euclidean distance between $\mathbf B_{i\cdot} \mathbf \Lambda^{1/2}$ and $\mathbf B_{j \cdot} \mathbf \Lambda^{1/2}$ will be equal to $d(\mathbf x_i, \mathbf x_j)$.
In classical scaling, for a $k$-dimensional representation of the samples we use the first $k$ columns of $\mathbf B$  (those corresponding to the largest eigenvalues) scaled by the appropriate eigenvalue:
\begin{align}
\mathbf M = \mathbf B_{\cdot, 1:k} \mathbf \Lambda^{1/2}_{1:k, 1:k}\label{Eq:mds-embedding}
\end{align}

Although classical multi-dimensional scaling is usually understood as simply giving a representation of the samples in the embedding space, we can define the necessary function as the one that maps supplemental points (i.e., points that were not used to create the MDS solution) to the embedding space.
To create such a function, we consider the problem of adding a point to the embedding space so as to have the distances in the embedding space match the distances defined by $d$.
The solution to this problem is derived in \citet{Gower1968-sb}.
Let $f: \R^p \to \R^k$ be the function that takes a new point and maps it to the embedding space.
If $\mathbf z \in \R^p$, then
\begin{align}
  f : \mathbf z \mapsto \frac{1}{2} \mathbf \Lambda_{1:k,1:k}^{-1} \mathbf M^T \mathbf a. \label{Eq:mds-map}
\end{align}
where $\mathbf a \in \R^n$ has elements
\begin{align}
  a_i = \left(-\frac{1}{2} \mathbf C_n \mathbf \Delta \mathbf C_n^T\right)_{ii} - d\left(\mathbf x_i, \mathbf z\right)^2\label{Eq:a-def}.
\end{align}

Note that adding a supplemental point in this way is not the same as making a new embedding based on $n+1$ data points: the positions of the original data points in the embedding space remain unchanged.
This is analogous to our interpretation of biplot axes for PCA: the biplot axis for the $j$th variable is the $j$th column of the Jacobian of a map taking a new point to the principal subspace, assuming that the principal subspace is fixed.

\subsection{Local biplot axes for differentiable distances}

Now that we have a function associated with a multi-dimensional scaling solution that maps from data space to embedding space, we are ready to generalize PCA biplots to multi-dimensional scaling.
Recall that the PCA biplot axis for the $j$th variable was given by the $j$th column of $J_g$, where $g$ was the map from data space to embedding space in PCA.
For multi-dimensional scaling, if $f$ is the map defined in (\ref{Eq:mds-map}), $J_f(\mathbf z)$ is a function of $\mathbf z$.
Therefore, we cannot have just one set of biplot axes describing the entire plot.
There will instead be one set of biplot axes for each point in the space of the original variables.

We make the following definition:
\theoremstyle{definition}
\begin{definition}{Local biplot axes.}

  
  Let $\mathbf X \in \R^{n \times p}$ be a data matrix, and recall that $\mathbf x_i$, $i = 1,\ldots, n$ are column vectors corresponding to the rows of $\mathbf X$.
  Let $d : \R^p \times \R^p \to \R^+$ be a distance function, and let $f$ be the function, defined in \ref{Eq:mds-map}, that maps supplemental points to the embedding space defined by classical scaling on $(\mathbf X, d)$.
  Denote by $d_{\mathbf y}$ the restriction of $d$ to $\{\mathbf y\} \times \R^p$, so that $d_{\mathbf y}: \{ \mathbf y\} \times \R^p \to \R^+$ is defined by $\mathbf x \mapsto d(\mathbf y, \mathbf x)$, and suppose that $d_{\mathbf x_i}$, $i = 1,\ldots, n$ has partial derivatives.
The {\em local biplot axes} for $\mathbf z \in \R^p$ are given by
\begin{align}
  LB(\mathbf z) &= (J_f(\mathbf z))^T \\
  &= \frac{1}{2} \begin{pmatrix}
    \frac{\partial d_{\mathbf x_1}(\mathbf z)^2}{\partial z_1} & \cdots \frac{\partial d_{\mathbf x_n}(\mathbf z)^2}{\partial z_1}\\
    \vdots &  \vdots \\
    \frac{\partial d_{\mathbf x_1}(\mathbf z)^2}{\partial z_p} & \cdots \frac{\partial d_{\mathbf x_n}(\mathbf z)^2}{\partial z_p}\\
  \end{pmatrix} \mathbf M \mathbf \Lambda_{1:k,1:k}^{-1}\\
  &= \begin{pmatrix}
    d_{\mathbf x_1}(\mathbf z)\frac{ \partial d_{\mathbf x_1}(\mathbf z)}{\partial z_1} & \cdots d_{\mathbf x_n}(\mathbf z)\frac{\partial d_{\mathbf x_n}(\mathbf z)}{\partial z_1}\\
    \vdots &  \vdots \\
    d_{\mathbf x_1}(\mathbf z)\frac{\partial d_{\mathbf x_1}(\mathbf z)}{\partial z_p} & \cdots d_{\mathbf x_n}(\mathbf z)\frac{\partial d_{\mathbf x_n}(\mathbf z)}{\partial z_p}\\
  \end{pmatrix} \mathbf M \mathbf \Lambda_{1:k,1:k}^{-1} \\
  &= \begin{pmatrix}
    \frac{\partial d_{\mathbf x_1}(\mathbf z)}{\partial z_1} & \cdots \frac{\partial d_{\mathbf x_n}(\mathbf z)}{\partial z_1}\\
    \vdots &  \vdots \\
    \frac{\partial d_{\mathbf x_1}(\mathbf z)}{\partial z_p} & \cdots \frac{\partial d_{\mathbf x_n}(\mathbf z)}{\partial z_p}\\
  \end{pmatrix} \text{diag}\left(\left(d_{\mathbf x_1}(\mathbf z), \ldots, d_{\mathbf x_n}(\mathbf z)\right)  \right)
  \mathbf M \mathbf \Lambda_{1:k,1:k}^{-1}
\end{align}
where $\mathbf M$ and $\mathbf \Lambda$ are as defined in equations \ref{Eq:squared-dist}-\ref{Eq:mds-embedding}.
The $j$th row of $LB(\mathbf z)$ is the local biplot axis for the $j$th variable.
\end{definition}
We define $LB(\mathbf z)$ as the transpose of the Jacobian matrix so that $LB(\mathbf z)$ is analogous to the matrix $\mathbf V_{\cdot, 1:k}$ used in the definition of PCA biplots earlier in this section.
We of course could have defined it the other way and taken the local biplot axis for the $j$th variable to be the $j$th column of the Jacobian matrix.

\subsection{Local biplot axes for non-smooth distances}

Although the definition above relies on the existence of the Jacobian and therefore requires the distance to be differentiable, we can modify our definition of local biplot axes to accommodate non-differentiable or discontinuous distances.
If the issue with the distance is that the left and right limits in the definition of the derivative exist but do not match, we use either the left or the right limit and call them negative or positive local biplot axes.
If the issue is that the limits diverge, as they might with a discontinuous distance, we use a discrete approximation of the derivative with discretization $\varepsilon$ that the user specifies and call them the $\varepsilon$-negative or $\varepsilon$-positive local biplot axes.
\theoremstyle{definition}
\begin{definition}{Local biplot axes for discontinuous distances.}  
  Let $\mathbf X \in \R^{n \times p}$ be a data matrix whose rows are denoted $\mathbf x_i$, $d : \R^p \times \R^p \to \R^+$ be a distance function, and let $f$ be the function, defined in \ref{Eq:mds-map}, that maps supplemental points to the embedding space defined by classical scaling on $(\mathbf X, d)$, and let $\varepsilon \in \R^+$.
  Denote by $d_{\mathbf y}$ the restriction of $d$ to $\{\mathbf y\} \times \R^p$, so that $d_{\mathbf y}(\mathbf x) = d(\mathbf y, \mathbf x)$.
The {\em $\varepsilon$-positive local biplot axes} for $\mathbf z \in \R^p$ are given by
\begin{align}
  LB^{\varepsilon,+}(\mathbf z) &= \mathbf F^{\varepsilon,+}(\mathbf z) \text{diag}\left(\left(d_{\mathbf x_1}(\mathbf z), \ldots, d_{\mathbf x_n}(\mathbf z)\right) \right) \mathbf M \mathbf \Lambda_{1:k,1:k}^{-1}
  \label{Eq:lb-pos-eps}
\end{align}
where $\mathbf F^{\varepsilon, +}(\mathbf z) \in \R^{p \times n}$, with
\begin{align}
(\mathbf F^{\varepsilon,+}(\mathbf z))_{ji} &= \frac{d_{\mathbf x_i}(\mathbf z + \varepsilon \mathbf e_j) - d_{\mathbf x_i}(\mathbf z)}{\varepsilon}
\end{align}
Analogously, still taking $\varepsilon > 0$, the {\em $\varepsilon$-negative local biplot axes} for $\mathbf z \in \R^p$ are given by
\begin{align}
  LB^{\varepsilon,-}(\mathbf z) &= \mathbf F^{\varepsilon, -}(\mathbf z) \text{diag}\left(\left(d_{\mathbf x_1}(\mathbf z), \ldots, d_{\mathbf x_n}(\mathbf z)\right) \right) \mathbf M \mathbf \Lambda_{1:k,1:k}^{-1}
  \label{Eq:lb-neg-eps}
\end{align}
where $\mathbf F^{\varepsilon, +}(\mathbf z) \in \R^{p \times n}$, with
\begin{align}
(\mathbf F^{\varepsilon,-}(\mathbf z))_{ji} &= \frac{d_{\mathbf x_i}(\mathbf z) - d_{\mathbf x_i}(\mathbf z - \varepsilon \mathbf e_j)}{\varepsilon}
\end{align}

\end{definition}

Although at first read, the $\varepsilon$ might seem to be an unpleasant hack, it has a reasonable interpretation.
In many situations there is a minimum amount by which a variable can be perturbed.
For instance, with count-valued data, a variable can be perturbed by no less than 1.
The interpretation of the local biplot axes in that case is that they describe how a supplemental point would react to perturbation of a variable by the minimum unit.

Note that by definition, if we use (\ref{Eq:lb-pos-eps}) or (\ref{Eq:lb-neg-eps}) on a differentiable distance, we will have $\lim_{\varepsilon \downarrow 0} LB^{\varepsilon, +}(\mathbf z) = \lim_{\varepsilon \downarrow 0} LB^{\varepsilon, -}(\mathbf z) = LB(\mathbf z)$ and so the two definitions are consistent.

Finally, we can have distances for which the relevant derivatives do not exist because although the left and right limits in the definition of the derivative exist, they do not match.
One example is the Manhattan distance $d_M(\mathbf x, \mathbf y) = \sum_{j=1}^p |x_j - y_j|$.
In that case, we can make the definition:
\theoremstyle{definition}
\begin{definition}{Local biplot axes for continuous non-differentiable distances.}  
  Let $\mathbf X \in \R^{n \times p}$ be a data matrix whose rows are denoted $\mathbf x_i$, $d : \R^p \times \R^p \to \R^+$ be a distance function, and let $f$ be the function, defined in \ref{Eq:mds-map}, that maps supplemental points to the embedding space defined by classical scaling on $(\mathbf X, d)$, and let $\delta \in \R^+$.
  Denote by $d_{\mathbf y}$ the restriction of $d$ to $\{\mathbf y\} \times \R^p$, so that $d_{\mathbf y}(\mathbf x) = d(\mathbf y, \mathbf x)$.
  Suppose that $\lim_{\varepsilon \downarrow 0} \frac{d_{\mathbf x_i}(\mathbf z + \varepsilon \mathbf e_j) - d_{\mathbf x_i}(\mathbf z)}{\varepsilon}$ and $\lim_{\varepsilon \downarrow 0} \frac{d_{\mathbf x_i}(\mathbf z) - d_{\mathbf x_i}(\mathbf z - \varepsilon e_j)}{\varepsilon}$ exist for all $i = 1,\ldots, n$, $j = 1,\ldots,p$, $\mathbf z \in \R^p$.

The {\em positive local biplot axes} for $\mathbf z \in \R^p$ are given by
\begin{align}
  LB^{+}(\mathbf z) &= \lim_{\varepsilon \downarrow 0}\mathbf F^{\varepsilon,+}(\mathbf z) \text{diag}\left(\left(d_{\mathbf x_1}(\mathbf z), \ldots, d_{\mathbf x_n}(\mathbf z)\right) \right) \mathbf M \mathbf \Lambda_{1:k,1:k}^{-1}
  \label{Eq:lb-pos}
\end{align}
where $\mathbf F^{\varepsilon, +}(\mathbf z) \in \R^{p \times n}$, with
\begin{align}
(\mathbf F^{\varepsilon,+}(\mathbf z))_{ji} &= \frac{d_{\mathbf x_i}(\mathbf z + \varepsilon \mathbf e_j) - d_{\mathbf x_i}(\mathbf z)}{\varepsilon}
\end{align}
Analogously, the {\em negative local biplot axes} for $\mathbf z \in \R^p$ are given by
\begin{align}
  LB^{-}(\mathbf z) &= \lim_{\varepsilon \downarrow 0}\mathbf F^{\varepsilon, -}(\mathbf z) \text{diag}\left(\left(d_{\mathbf x_1}(\mathbf z), \ldots, d_{\mathbf x_n}(\mathbf z)\right) \right) \mathbf M \mathbf \Lambda_{1:k,1:k}^{-1}
  \label{Eq:lb-neg}
\end{align}
where $\mathbf F^{\varepsilon, +}(\mathbf z) \in \R^{p \times n}$, with
\begin{align}
(\mathbf F^{\varepsilon,-}(\mathbf z))_{ji} &= \frac{d_{\mathbf x_i}(\mathbf z) - d_{\mathbf x_i}(\mathbf z - \varepsilon \mathbf e_j)}{\varepsilon}
\end{align}
As with the definition for discontinuous distances, if we apply this definition to a distance that is differentiable, we will have $LB^+(\mathbf z) = LB^-(\mathbf z) = LB(\mathbf z)$.

\end{definition}

\subsection{Properties of local biplot axes}

We next show some simple properties of local biplot axes that illustrate how they are related to other methods and how they allow us to interpret the MDS embedding space.
Proofs are provided in the Appendix.

\subsubsection{Equivalence of PCA biplot axes and local biplot axes for Euclidean distances}

The first property of local biplot axes has to do with the relationship between the principal components of $\mathbf X$ and the local biplot axes that result when we perform MDS on $\mathbf X$ with the Euclidean distance.

\begin{thm}
  Equivalence of PCA axes and local biplot axes.
  Let $\mathbf X \in \R^{n \times p}$ have centered columns, $d : \R^p \times \R^p \to \R$ be the Euclidean distance function, $d(\mathbf x, \mathbf y) = \left[\sum_{i = 1}^p (x_i - y_i)^2\right]^{1/2}$.
  Let the singular value decomposition of $\mathbf X$ be $\mathbf X = \mathbf U \mathbf \Lambda \mathbf V^T$, where $\mathbf U \in \R^{n \times k}$, $\mathbf \Lambda \in \R^{k \times k}$, $\mathbf V \in \R^{p \times k}$, $\mathbf U^T \mathbf U = \mathbf V^T \mathbf V = \mathbf I_k$.
  Then for any $\mathbf z \in \R^p$, the local biplot axes for classical multi-dimensional scaling of $\mathbf X$ with distance $d$ at $\mathbf z$ are given by $LB(\mathbf z) = \mathbf V_{\cdot, 1:k}$.
  \label{Thm:pca-mds}
\end{thm}

This theorem tells us that if we perform MDS with the Euclidean distance, the local biplot axes are constant and they are the same as the principal axes.
It can be thought of as dual to the classic result about the relationship between PCA and multi-dimensional scaling with the Euclidean distance in \citet{Gower1966-hx}: the result above is about the variables, and the classic result is about the samples.

\subsubsection{Relationship between gPCA axes and local biplot axes for generalized Euclidean distances}

The relationship between local biplot axes and principal components is not restricted to the standard Euclidean distance.
For any generalized Euclidean distance $d_{\mathbf Q}(\mathbf x,\mathbf y) = \sqrt{(\mathbf x-\mathbf y)^T \mathbf Q(\mathbf x-\mathbf y)}$, with $\mathbf Q \succ 0$, the local biplot axes will be constant on the data space and will be related to the generalized eigendecomposition or generalized principal components analysis (gPCA) of $\mathbf X$.
The interested reader can see \citet{Holmes2008-op} for a review.
Briefly, gPCA is defined on a triple $(\mathbf X, \mathbf Q, \mathbf D)$, where $\mathbf X \in \R^{n \times p}$, $0 \prec \mathbf Q \in \R^{p \times p}$, $0 \prec \mathbf D \in \R^{n \times n}$.
The generalization in gPCA can be interpreted either as generalizing the noise model in the probabilistic formulation of PCA to one in which the errors have a matrix normal distribution with covariance $\mathbf D^{-1} \otimes \mathbf Q^{-1}$ \citep{allen2014generalized}, or generalizing from standard inner product spaces on the rows and columns of $\mathbf X$ to inner product spaces defined by $\mathbf Q$ and $\mathbf D$ \citep{Holmes2008-op}.

The local biplot axes for MDS with $d_{\mathbf Q}$ are related to the generalized principal axes for the triple $(\mathbf X, \mathbf Q, \mathbf I_n)$.
The definition of generalized principal axes is as follows:
\begin{definition}{Generalized principal axes.}
  Consider gPCA of the triple $(\mathbf X, \mathbf Q, \mathbf D)$, with $\mathbf X \in \R^{n \times p}$ , $0 \prec \mathbf Q \in \R^{p \times p}$, $0 \prec \mathbf D \in \R^{n \times n}$.
  Let $\mathbf V \in \R^{p \times k}$ and $\mathbf \Lambda \in \R^{k \times k}$ be matrices satisfying
  \begin{align}
    \mathbf X^T \mathbf D \mathbf X \mathbf Q \mathbf V = \mathbf V \mathbf \Lambda, \quad \mathbf V^T \mathbf Q \mathbf V = \mathbf I\label{Eq:g-eig-def}
  \end{align}
  with $\lambda_{ii} \ge \lambda_{jj}$ iff $i < j$.
  The {\em generalized principal axes} for gPCA on the triple $(\mathbf X, \mathbf Q, \mathbf D)$ are given by $\mathbf V \mathbf \Lambda^{1/2}$, and the {\em normalized generalized principal axes} are simply $\mathbf V$.
  \label{Def:generalized-principal-axes}
\end{definition}

The generalized principal axes/generalized eigenvectors exist and can be understood in terms of the more familiar standard eigendecomposition of $\mathbf X \mathbf Q \mathbf X^T$ as follows:
\begin{thm}
  If the eigendecomposition of $\mathbf X \mathbf Q \mathbf X^T$ is $\mathbf{\tilde V}  \mathbf{\tilde \Lambda} \mathbf{\tilde V}^T$, then $\mathbf V = \mathbf Q^{-1/2} \mathbf{\tilde V}$ and $\mathbf \Lambda = \mathbf {\tilde \Lambda}$ satisfy the equations (\ref{Eq:g-eig-def}).
  \label{Thm:gpca-existence}
\end{thm}

Once we have defined generalized eigenvectors/generalized principal axes, we can write down the relationship between the generalized principal axes and the local biplot axes for a generalized Euclidean distance:
\begin{thm}
  Local biplot axes for generalized Euclidean distances.
  Let $\mathbf X \in \R^{n \times p}$, $\mathbf Q \in \R^{p \times p}$, $\mathbf Q \succ 0$, $d_{\mathbf Q} : \R^p \times \R^p \to \R$ be a generalized Euclidean distance so that $d_{\mathbf Q}(\mathbf x, \mathbf y)= [(\mathbf x - \mathbf y)^T \mathbf Q(\mathbf x - \mathbf y)]^{1/2}$.
  Let $\mathbf V$ be the normalized principal axes for gPCA on the triple $(\mathbf X, \mathbf Q, \mathbf I_n)$.
  Then the local biplot axes for classical multi-dimensional scaling of $\mathbf X$ with the distance $d_{\mathbf Q}$ at $\mathbf z$ are given by $LB(\mathbf z) = \mathbf Q \mathbf V$.
  \label{Thm:quad-dist}
\end{thm}
Note that standard PCA is gPCA on $(\mathbf X, \mathbf I_p, \mathbf I_n)$, and so Theorem \ref{Thm:pca-mds} is a special case of Theorem \ref{Thm:quad-dist}.

\subsubsection{Conditions under which constant local biplot axes imply generalized Euclidean distance}

After seeing that generalized Euclidean distances imply constant local biplot axes, a natural question is whether the converse holds: if the local biplot axes are constant, must the distance have been a generalized Euclidean distance?
The answer in general is no: given that $LB(\mathbf z) = \mathbf L$ for every $\mathbf z \in \R^p$, and assuming further that adding $\mathbf z$ does not require an additional embedding dimension, we still only have that $d(\mathbf x_i, \mathbf z) = d_{\mathbf L \mathbf L^T}(\mathbf x_i, \mathbf z)$ for $i = 1,\ldots, n$, $\mathbf z \in \R^p$.

\begin{thm}{Constant local biplot axes imply $d$ is a generalized Euclidean distance on $\{ \mathbf x_i \}_{i = 1,\ldots, n} \times \R^p$.}
  Let $\mathbf X \in \R^{n \times p}$, with $n > p$.
  Suppose the local biplot axes associated with the multi-dimensional scaling solution of $\mathbf X$ with the distance $d$ are $LB(\mathbf z) = \mathbf L$, with $\mathbf L \in \R^{p \times p}$ with $\text{rank}(\mathbf L) = p$ for any $\mathbf z \in \R^p$.
  Suppose that for multi-dimensional scaling of $\mathbf X$ with $d$, the full embedding space is of dimension $p$, and further suppose that any supplemental points $\mathbf z \in \R^p$ can be added to the space without requiring an additional embedding dimension.
    Then for any $\mathbf z \in \R^p$ and any $\mathbf x_i$, $i = 1,\ldots, n$ we have $d(\mathbf x_i, \mathbf z) = d_{\mathbf L \mathbf L^T}(\mathbf x_i, \mathbf z)$.
  
  \label{Thm:constant-axes-on-input}
\end{thm}

The limitation of the theorem above is that it does not apply to arbitrary pairs of points in the embedding space, only to pairs for which one member is one of the initial data points.
To see why constant local biplot axes are not enough to ensure that for any arbitrary points $\mathbf z_1, \mathbf z_2 \in \R^p$, that $d(\mathbf z_1, \mathbf z_2) = d_{\mathbf L \mathbf L^T}(\mathbf z_1, \mathbf z_2)$, consider the distance
\begin{align}
  d_{ET}(x,y)^2 = \begin{cases}
    (x - y)^2 & \text{min}(|x|, |y|) \le 1 \\
    \frac{1}{4}(x - y)^2 & \text{min}(|x|, |y|) > 1
  \end{cases}.
\end{align}
The mnemonic is that ET stands for ``express train'': if we think of the distance as being time along a rail route, trips from anywhere to any point in $[-1,1]$ take time equal to distance, while trips between two points outside of $[-1,1]$ take time equal to half the distance.

If we imagine performing multi-dimensional scaling with $d_{ET}$ and $\mathbf X = \begin{pmatrix} .5 \\ -.5 \end{pmatrix}$, the MDS solution would give us a one-dimensional space with $\mathbf x_1$ embedded at $.5$ and $\mathbf x_2$ embedded at $-.5$ (or vice versa).
Additionally, the local biplot axis (axis singular because there is only one variable in this example) would be constant and equal to 1.
This is because the local biplot axes are defined by $d_{\mathbf x_i}(\mathbf z)$, and since both $\mathbf x_1$ and $\mathbf x_2$ have absolute value less than 1, the distance to be evaluated in the definition of LB axes is always the standard Euclidean distance.
But of course $d_{ET}$ is not the standard Euclidean distance and is indeed not a generalized Euclidean distance at all, and so constant local biplot axes cannot imply that the distance used for MDS was a generalized Euclidean distance.

However, if  we add the assumption that $d$ is homogeneous and translation invariant, constant local biplot axes do imply that $d$ is a generalized Euclidean distance on $\R^p \times \R^p$, not just on $\{ \mathbf x_i \}_{i = 1,\ldots, n} \times \R^p$.
Note that the assumption is not the same as restricting to the set of generalized Euclidean distances, as there any many distances that satisfy those properties that are not generalized Euclidean distances, e.g. the Minkowski distances.
In addition, note that the homogeneity and translation invariance assumptions do not rule out any generalized Euclidean distance as all generalized Euclidean distances are homogeneous and translation invariant.

\begin{thm}{Constant local biplot axes and homogeneous, translation-invariant $d$ imply $d$ is a generalized Euclidean distance on $\R^p \times \R^p$.}
  Suppose that $n \ge p$, $\mathbf X \in \R^{n \times p}$ satisfies $\text{rank}(\mathbf X) = p$, and that the local biplot axes associated with the multi-dimensional scaling solution of $\mathbf X$ with the distance $d$ are $LB(\mathbf z) = \mathbf L$ for $\mathbf L \in \R^{p \times p}$ with $\text{rank}(\mathbf L) = p$ for any $\mathbf z \in \R^p$.
  As in Theorem \ref{Thm:constant-axes-on-input}, assume that the full-dimensional MDS embedding space has dimension $p$, and that for any $\mathbf z \in \R^p$, $\mathbf z$ can be added to the initial embedding space without requiring the addition of any extra dimensions.
  
  Suppose further that $d$ satisfies translation invariance and homogeneity, so that for any $\alpha \in \R$ $\mathbf x, \mathbf y, \mathbf z \in \R^p$, $d(\alpha \mathbf x, \alpha \mathbf y) = |\alpha| d(\mathbf x, \mathbf y)$ and $d(\mathbf x + \mathbf z, \mathbf y + \mathbf z) = d(\mathbf x, \mathbf y)$.
  Then for any $\mathbf z_1, \mathbf z_2 \in \R^p$, $d(\mathbf z_1, \mathbf z_2) = d_{\mathbf L \mathbf L^T}(\mathbf z_1, \mathbf z_2)$.
  \label{Thm:constant-axes-with-dist-assumptions}
\end{thm}

These two results are useful for interpretation of the MDS embedding space in the following way:
In gPCA, the lower-dimensional representation of the samples is obtained by projecting the samples onto the gPCA axes, with the projection being with respect to the $\mathbf Q$ inner product.
If we are performing MDS with a generalized Euclidean distance, we can obtain the axes onto which the samples are being projected with Theorem \ref{Thm:quad-dist}.
In the more common situation where we are performing MDS with an arbitrary distance, if we see that the local biplot axes are approximately constant, Theorems \ref{Thm:quad-dist} and \ref{Thm:constant-axes-on-input} suggest that the distance can be approximated by a generalized Euclidean distance and that the relationship between the data space and the embedding space can be approximated as a projection onto the local biplot axes.

\subsubsection{Supplemental point located at centroid of local biplot axes for generalized Euclidean distance}

Our next result relates the position of a supplemental point in the MDS embedding space to the centroid of the local biplot axes scaled up or down according to the values of the the variables for the supplemental point.
As mentioned in the introduction, our local biplots are inspired by the non-linear biplots in \citet{Gower1988-fj}, and Gower gives a similar interpolation result for his non-linear biplots in the special case of distances that are decomposable by coordinate.
Before proceeding to our result for local biplots and generalized Euclidean distances, we give a restatement of Gower's result.
Our statement of the theorem is in terms of the embeddings of supplemental points corresponding to scaled standard basis vectors instead of non-linear biplot axes, but it amounts to the same thing.
\begin{thm}{Interpolation.}
  Let $\mathbf X \in \R^{n \times p}$, and perform MDS on $(\mathbf X, d)$ with $d$ such that $d^2(\mathbf x, \mathbf y) = \sum_{j=1}^p h(x_j, y_j)$.
  The rows of $\mathbf X$ are denoted $\mathbf x_i$, and the $ij$th elemont of $\mathbf X$ is $x_{ij}$.
  Let $\mathbf z = \sum_{j=1}^p \alpha_j \mathbf e_j$ be a supplemental point, let $\mathbf p_j$ be the embedding of $\alpha_j \mathbf e_j$ in MDS space: $\mathbf p_j := f(\alpha_j \mathbf e_j)$, where $f$ is the map from data space to embedding space as defined in \ref{Eq:mds-map}.
  Let $\mathbf c = \frac{1}{p} \sum_{j=1}^p \mathbf p_j$.
  The coordinates of $f(\mathbf z) = f(\sum_{j=1}^p \alpha_j \mathbf e_j)$ will then be $p \mathbf c - (p-1)f(\mathbf 0_p)$.
\label{Thm:gower-centroid}
\end{thm}
The points $\mathbf p_j$ are points on Gower's non-linear biplot axes, and the result shows that the embedding of a supplemental point can be found from the centroid $\mathbf c$ of the non-linear biplot axes for each variable represented in the supplemental point.

We have a similar interpolation result for local biplot axes.
In our case, instead of restricting the class of distances to those decomposable by variable, we restrict to generalized Euclidean distances.
\begin{thm}{Interpolation for generalized Euclidean distances.}
  Let $\mathbf Q \succ 0$ and $d : \R^p \times \R^p \to \R^+$ be defined as $d(\mathbf x, \mathbf y) = \left[ (\mathbf x - \mathbf y)^T \mathbf Q (\mathbf x - \mathbf y) \right]^{1/2}$.
  By Theorem \ref{Thm:quad-dist}, we know that $LB(\mathbf z) = \mathbf Q \mathbf V$ for some $\mathbf V \in \R^{p \times k}$ and any $\mathbf z \in \R^p$.
  Let $\mathbf c = \frac{1}{p} \sum_{j=1}^p LB(\mathbf z)_{\cdot, j}^T = \frac{1}{p} \sum_{j=1}^p \alpha_j (\mathbf Q \mathbf V)_{\cdot, j}^T$ be the signed and weighted centroid of the local biplot axes, with weights $\alpha_j$.
  The embedding of the supplemental point $\mathbf z$ will be $f(\mathbf z) = p \mathbf c$.
\label{Thm:lb-centroid}
\end{thm}
As before, we expect this result to hold approximately when the local biplot axes are approximately constant and the distance can be interpreted as approximately a generalized Euclidean distance.




\section{Results}

To illustrate the value of local biplots and to show what sorts of insights they can provide, we present local biplots on real and simulated datasets.

\subsection{Simulated data for phylogenetic distances}

To show how local biplots can help us interpret an MDS embedding of a set of samples, we create MDS embeddings and the associated local biplots for two different distances on a single dataset.
To make the comparison as straightforward as possible, we set up the simulated data so that the MDS embeddings of the samples with the different distances are approximately the same.
We will see that despite the similarity of the sample embeddings, the features used to create the embeddings are quite different for the two distances.

The distances we will use for the comparison are the Manhattan distance and weighted UniFrac \citep{lozupone2007quantitative}.
Recall that the Manhattan distance between points $\mathbf x, \mathbf y \in \R^p$ is $d_M(\mathbf x, \mathbf y) = \sum_{i=1}^p | x_i - y_i|$.
Because the Manhattan distance is continuous but is not differentiable everywhere, we use the positive local biplot axes, as defined in (\ref{Eq:lb-pos}).
The results are qualitatively similar for the negative local biplot axes.
There is a slight difference in that for the negative local biplot axes, many were exactly the same, which has to do with the fact that the data are bounded below at 0 and many data points lie exactly on that boundary.

Weighted UniFrac is a phylogenetically-informed distance commonly used to quantify the dissimilarity between bacterial communities.
If we have a phylogenetic tree $\mathcal T$ with $B$ branches, $p$ tips labeled $1,\ldots, p$ and vectors $\mathbf x, \mathbf y \in \R^p$, then the weighted UniFrac distance between $\mathbf x$ and $\mathbf y$ is $d_w(\mathbf x,\mathbf y) = \sum_{b=1}^B l_b \left | \sum_{i \in \text{desc}(b, \mathcal T)} x_i / \|\mathbf x \|_1 - \sum_{i \in \text{desc}(b, \mathcal T)} y_i / \|\mathbf y \|_1 \right|$ where $\text{desc}(b, \mathcal T) = \{i : \text{ tip $i$ descends from branch $b$} \}$ and $l_b$ is the length of branch $b$.

The idea behind the weighted UniFrac distance is that it takes into account similarity between the variables as measured by proximity on $\mathcal T$.
It does not treat all of the variables as equally dissimilar the way more standard distances do.
Weighted UniFrac is also the solution to an optimal transport problem \citep{evans2012phylogenetic}.

So that we can perform MDS with weighted UniFrac, we create a phylogenetic tree $\mathcal T$ with $p$ leaves and a matrix $\mathbf X \in \R^{n \times p}$.
$\mathcal T$ is a balanced phylogenetic tree with $p$ leaves, labeled $1,\ldots, p$.
Let $\mathbf 1_{shallow} \in \{0,1\}^p$ be such that if $i$ and $j$ represent sister taxa in $\mathcal T$, exactly one of $(\mathbf 1_{shallow})_{i}, (\mathbf 1_{shallow})_j$ is equal to 1.
Call the two children of the root of $\mathcal T$ $l$ and $r$.
Let $\mathbf 1_{deep}\in \{0,1\}^p$ be such that $(\mathbf 1_{deep})_i = 1$ if taxon $i$ descends from $l$ and 0 otherwise.
Let $\mathbf 1_a = (\mathbf 1_{n/2}^T, \mathbf 0_{n/2}^T)$.
We then define
\begin{align}
  \mathbf A = &c_1 (\mathbf 1_a \mathbf 1_{shallow}^T + (\mathbf 1_n - \mathbf 1_a) (\mathbf 1_p - \mathbf 1_{shallow})^T)\, \odot \\
  \quad&\exp(c_2 (\mathbf 1_a \mathbf 1_{deep}^T + (\mathbf 1_n - \mathbf 1_a) (\mathbf 1_p - \mathbf 1_{deep})^T)) \\
  x_{ij} \sim &\text{Double Pois}(a_{ij}, s)
\end{align}
where $c_1, c_2$ are constants, $\mathbf 1_n$ indicates the $n$-vector containing all 1's, $\mathbf 1_{shallow}$ and $\mathbf 1_{deep}$ are as defined above, $\exp$ applied to a vector indicates the element-wise operation, and $\odot$ indicates the element-wise product.

The idea behind this setup is that there are two groups of samples, those for which $(\mathbf 1_a)_i = 1$ and those for which $(\mathbf 1_a)_i = 0$.
There are are two differences between these groups: the first is a mass difference from one half of the tree to the other, and the other an exclusion effect, in which for each pair of sister taxa, if one is present the other is absent.
These two effects can be seen in Figure \ref{Fig:simulation-setup}, which provides a visualization of the tree $\mathcal T$ and the data matrix $\mathbf X$.

\begin{figure}
  \begin{center}
    \includegraphics[width=\textwidth]{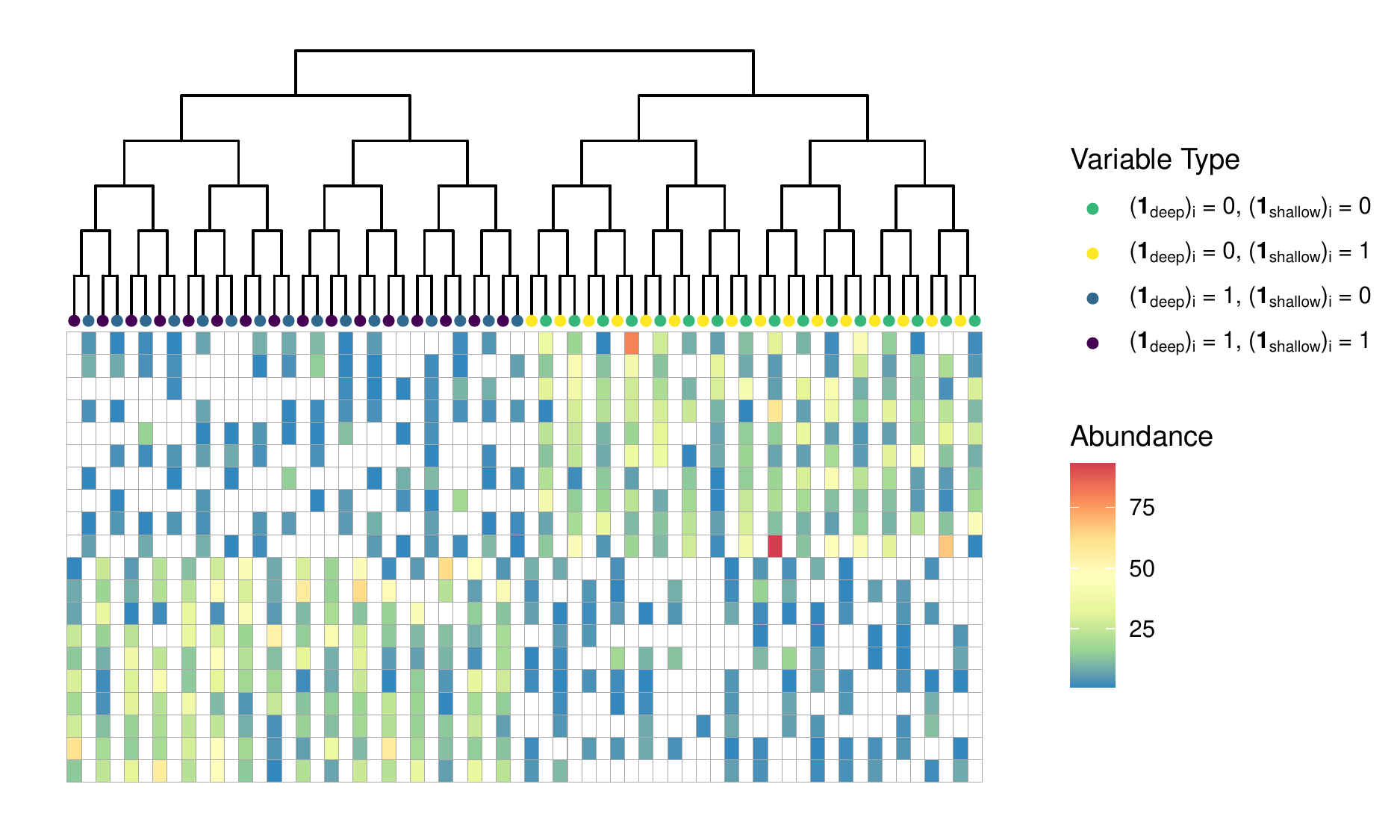}
  \end{center}
  \caption{Phylogenetic tree (top) and simulated data (bottom).
    White boxes indicate 0 abundance in the data matrix $\mathbf X$, and other colors from blue to red indicate low to high non-zero abundances.
    Samples are encoded in rows, and variables in the columns.
    Each variable corresponds to one tip of the tree, and the colored dots at the tips of the tree represent the variables.
    The variable colors are carried through to Figure \ref{Fig:biplot-simulation}.
    Note that the first group of ten samples in the top rows has a higher abundance of variables on the right-hand side of the tree than the left-hand side of the tree, while the reverse is true for the second group of the ten samples.
  In addition, the first group of ten samples contains only representatives of the blue and green variables, while the second group of ten samples contains only representatives of the purple and yellow variables.}
  \label{Fig:simulation-setup}
\end{figure}

Our first step is to create MDS embeddings of the samples in the rows of $\mathbf X$ using either the Manhattan distance or weighted UniFrac.
The results of that embedding are shown in Figure \ref{Fig:mds-simulation}, and as expected, the MDS representation of the samples separates the samples for which $(\mathbf 1_a)_i = 0$ from those for which $(\mathbf 1_a)_i = 1$ along the first axis.
Without local biplot axes, this is as far as MDS takes us: we see that in each case, the MDS embedding clusters the samples into the same two groups.
We know in principle that weighted UniFrac uses the tree and the Manhattan distance does not, but the MDS plot on its own provides no insight into the features of the data that go into the MDS embeddings.

\begin{figure}
  \includegraphics[width=\textwidth]{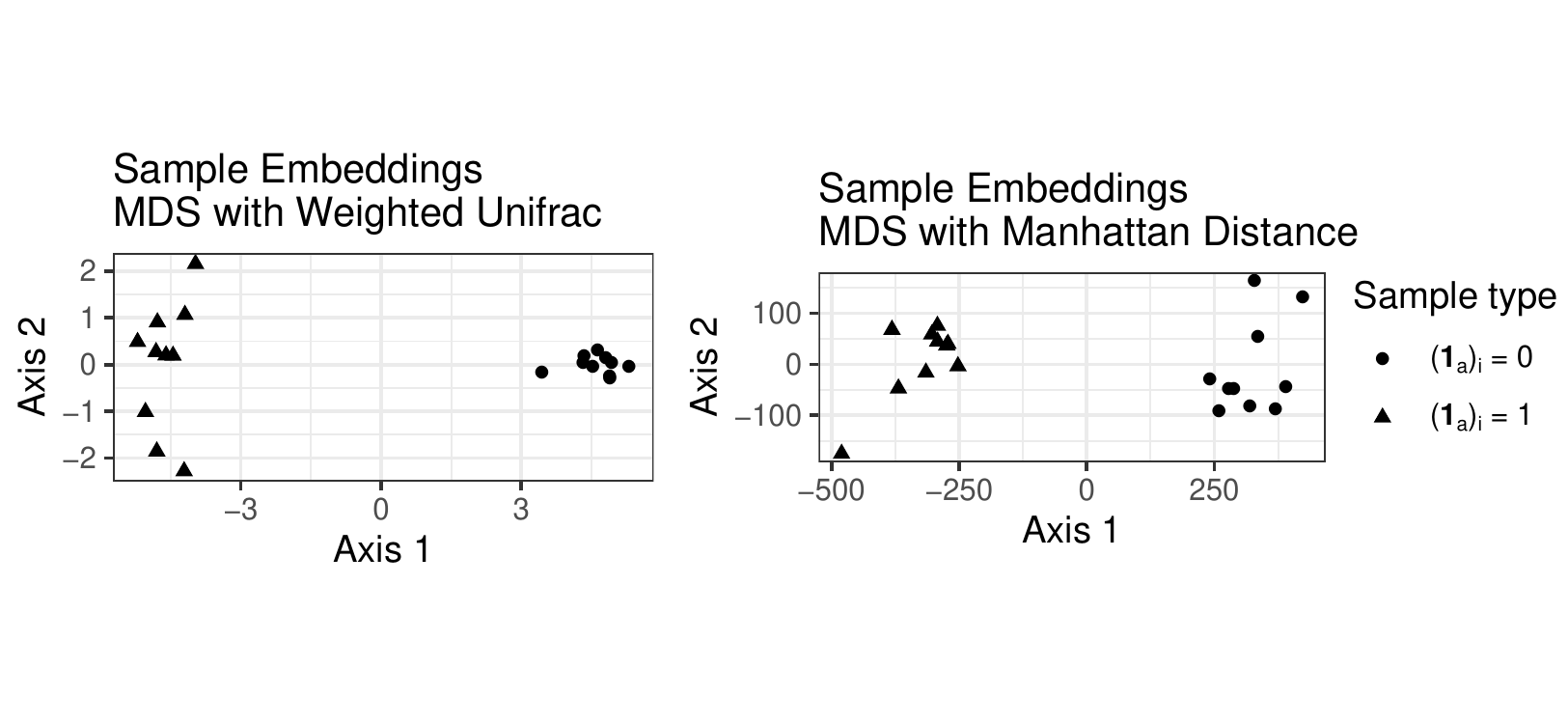}
  \caption{Classical multi-dimensional scaling representation of the samples in the simulated dataset using the weighted UniFrac distance (left) and Manhattan distance (right).
    Each point represents a sample, and shape represents which group the sample came from.
    The representations are not identical, but in each case MDS gives two clusters corresponding to the two groups we simulated from.}
  \label{Fig:mds-simulation}
\end{figure}

Once we have the MDS embeddings for weighted UniFrac and for the Manhattan distance on this dataset, we can add the local biplot axes to visualize how the variables relate to the two MDS embedding spaces.
We start with the local biplot axes for the Manhattan distance, which are shown in the bottom panel of Figure \ref{Fig:biplot-simulation}.
The LB axes were computed at the sample points.
We see that for the Manhattan distance, the variables with positive values for the first local biplot axes are those for which $(\mathbf 1_{shallow})_i = 1$, and the variables with negative values for the first local biplot axes are those for which $(\mathbf 1_{shallow})_i = 0$.
The local biplot axes are approximately the same everywhere in the space, and so we can interpret the MDS/Manhattan distance embedding along the first axis as being approximately a projection onto a vector describing the contrast between variables with values of 0 vs. 1 for $\mathbf 1_{shallow}$.

\begin{figure}
  \includegraphics[width=\textwidth]{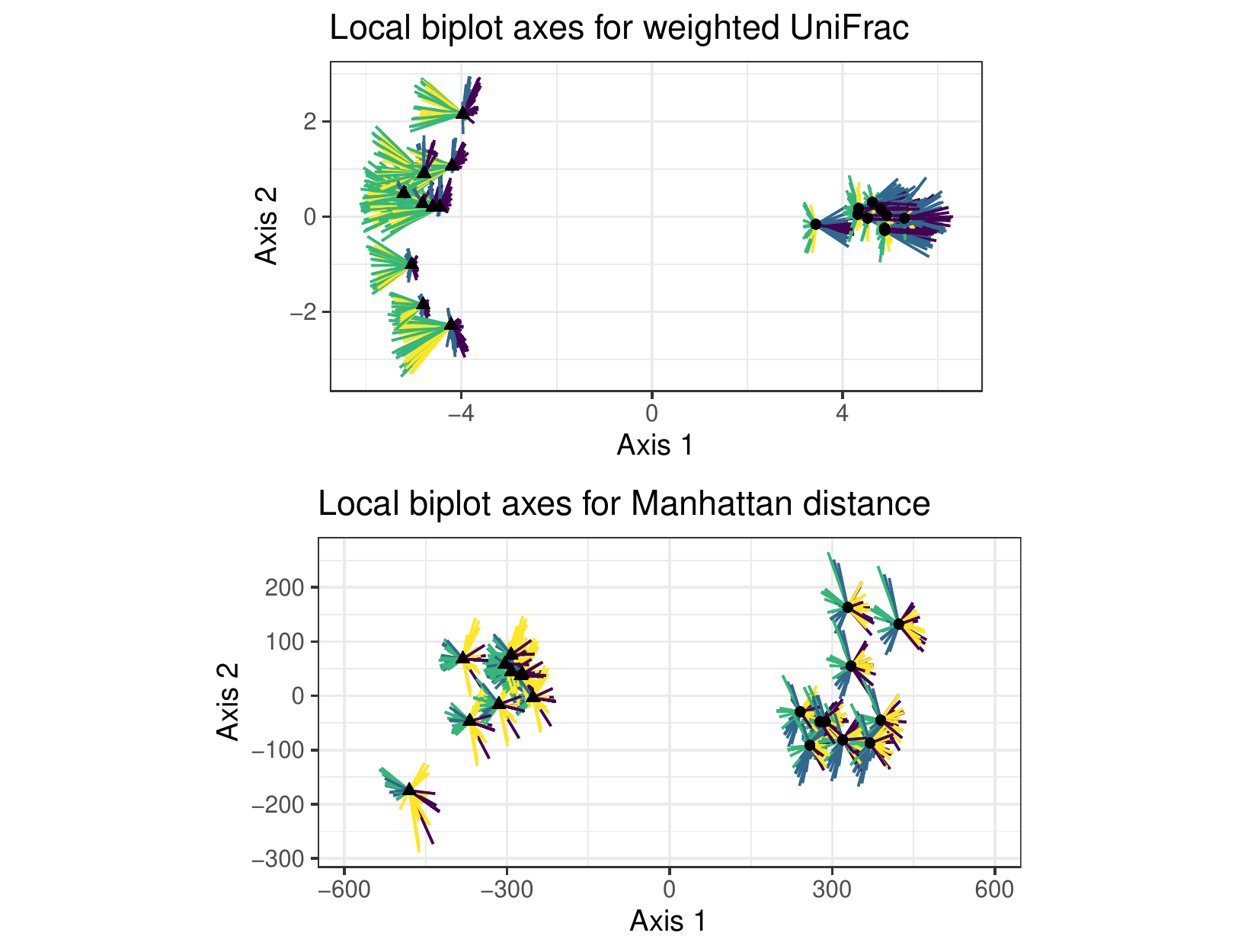}
  \caption{Local biplot axes for MDS/weighted UniFrac (top) and MDS/Manhattan distance (bottom).
    Black circles or triangles represent the sample embeddings.
    A segment connected to a sample point represents a local biplot axis for one variable at that sample point.
    Color represents variable type, and matches the colors of the points on the tips of the trees in Figure \ref{Fig:simulation-setup}.
    Blue/purple vs. yellow/green represent variables corresponding to $(\mathbf 1_{deep})_j = 0$ vs. $(\mathbf 1_{deep})_j = 1$, and blue/green vs. purple/yellow represent $(\mathbf 1_{shallow})_j = 0$ vs $(\mathbf 1_{shallow})_j = 1$.
    Note that the LB axes for weighted UniFrac have green/yellow axes with negative values on the first axis and blue/purple axes with positive values on the first axis.
    The situation is different for the Manhattan distance: there blue/green axes have negative first axis values, and yellow/purple have positive first axis values.
  }
  \label{Fig:biplot-simulation}
\end{figure}

We then turn to the local biplot axes for MDS with the weighted UniFrac distance to visualize the relationship between the variables and the MDS embedding space for weighted UniFrac.
The local biplot axes plotted were again computed at the sample points and are shown in Figure \ref{Fig:biplot-simulation}.
In contrast to the local biplot axes for MDS with the Manhattan distance, the variables with positive values on the first local biplot axes are those for which $(\mathbf 1_{deep})_i = 1$. The variables with negative values on the first local biplot axes are those for which $(\mathbf 1_{deep})_i = 0$.
The LB axes here are more variable than the LB axes for the Manhattan distance, but the signs of the first LB axes are consistent, and so we can interpret the first axis as being approximately a projection onto a vector describing the contrast between variables with value of 0 vs. 1 for $\mathbf 1_{deep}$.

To see how different the local biplot axes are from other proposals for explaining the MDS embedding space, we create correlation biplots for MDS with weighted UniFrac and MDS with the Manhattan distance.
The correlation biplot axes for these two embeddings are shown in the left and right panels of Figure \ref{Fig:correlation-biplot}, respectively.
In contrast to the local biplot axes for these two embeddings, the correlation biplot axes for weighted UniFrac and the Manhattan distance are very similar to each other.
In each case, the variables for which $(\mathbf 1_{shallow})_i = 0$ have positive values, while those for which $(\mathbf 1_{shallow})_i = 1$ have negative values.
Although this is one explanation of the differences between the groups, it is of course not the only one, and in particular it is not the one that is actually used by weighted UniFrac.
The reason the two sets of correlation biplot axes are so similar is that given the embeddings, the correlation biplot doesn't know anything about the distance used, and knows nothing about the relationship between the original data space and the embedding space.
In our example, the embeddings with weighted UniFrac and the Manhattan distance were very similar, and so the correlation biplot axes were also similar.

\begin{figure}
  \includegraphics[width=\textwidth]{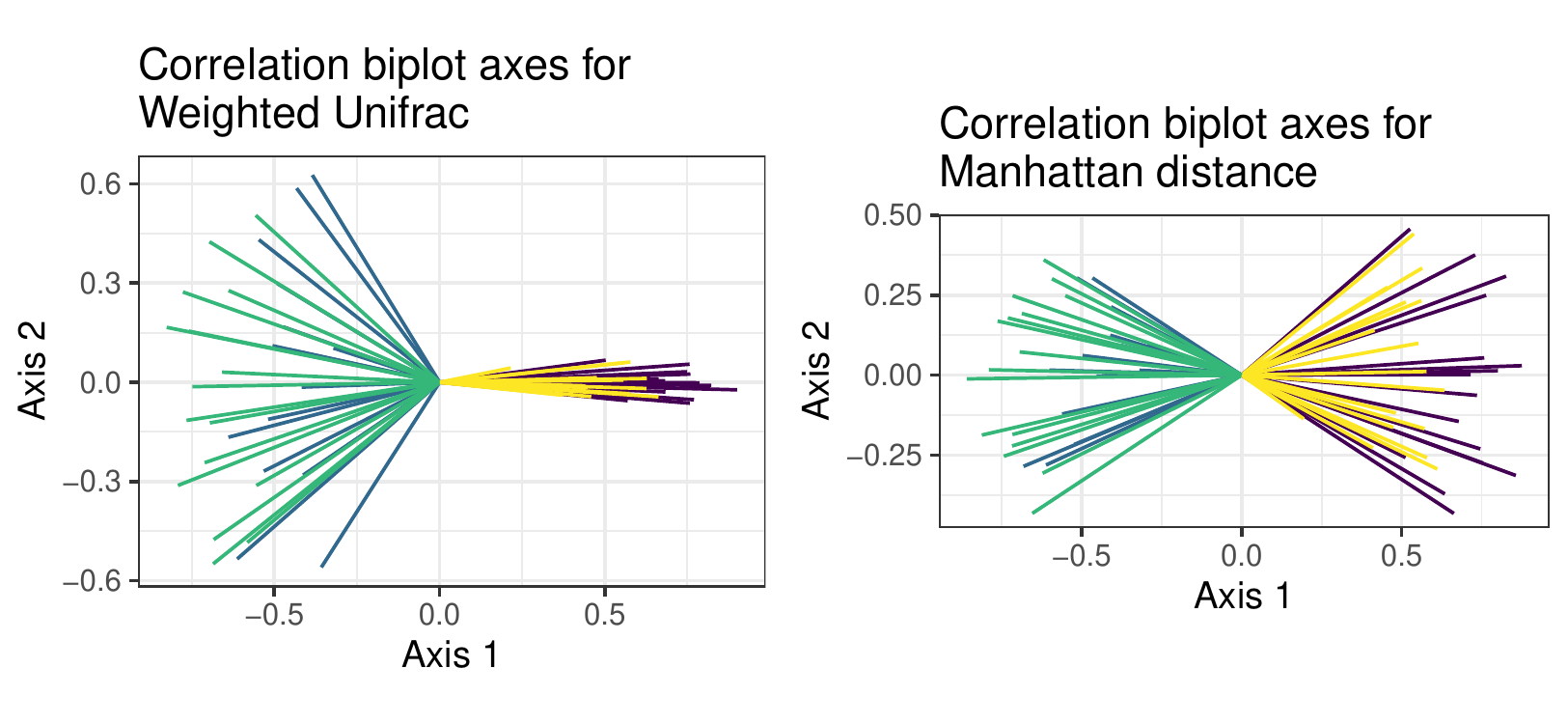}
  \caption{Correlation biplots for weighted UniFrac (left) and Manhattan distance (right).
    Each segment is a biplot axis.
    Color represents variable type, and matches the colors of the points on the tips of the trees in Figure \ref{Fig:simulation-setup}.
    Blue/purple vs. yellow/green represent variables corresponding to $(\mathbf 1_{deep})_j = 0$ vs. $(\mathbf 1_{deep})_j = 1$, and blue/green vs. purple/yellow represent $(\mathbf 1_{shallow})_j = 0$ vs $(\mathbf 1_{shallow})_j = 1$.
    Note the similarity of the correlation biplot axes for the two MDS representations, in contrast to the LB axes in Figure \ref{Fig:biplot-simulation}, which are quite different for weighted UniFrac vs. the Manhattan distance.
   }
  \label{Fig:correlation-biplot}
\end{figure}

In this simulation, and generally when we have $p > n$, there is more than one  way to separate two groups of samples.
The local biplot representation in this simulation shows how two different distances implicitly use different sets of features to separate the two groups: relative abundance of the two halves of the tree in for weighted UniFrac, and relative abundance of alternating sets of leaves on the tree for the Manhattan distance.
This is information that we might have expected based on the known properties of the distances (weighted UniFrac uses the tree and Manhattan distance doesn't), but that is not directly available otherwise.
We see that the local biplot axes give us much more insight into how the variables are used than the correlation biplot axes do.

This sort of information about how the variables relate to the MDS embedding space is particularly important if the distance was designed to incorporate the analyst's intuition about important features of the data for the particular problem.
In the case of weighted UniFrac, the intuition is that the tree provides important information, and an explanation of the difference between the two groups as being over- or under-represented in one half of the tree is more useful than an explanation of the difference between the two groups as being the presence or absence of half of the species.
Of course, the opposite could also be true: we could be more interested in the exclusion effect and try to design a distance that emphasizes that instead of the deep splits in the tree.
The local biplot axes would again be more useful than correlation biplot axes in that situation, because they describe the relationship between the data space and the embedding space instead of simply describing marginal relationships between the variables and the embeddings.

\subsection{Real data for phylogenetic distances}

To illustrate the utility of local biplots on a real dataset, we analyze data from a study of the effect of antibiotics on the gut microbiome initially described in \citet{dethlefsen2011incomplete}.
In this study, three individuals were given two courses of the antibiotic Ciprofloxacin, and stool samples were taken in the days and weeks before, during, and after each course of the antibiotic.
The composition of the bacterial communities in each of these samples was analyzed using 16S rRNA sequencing \citep{davidson2018microbiome}, and the resulting dataset describes the abundances of 1651 bacterial taxa in each of 162 samples (between 52 and 56 samples were taken per individual).
In addition, the taxa were mapped to a reference phylogenetic tree \citep{quast2013silva}, giving us the evolutionary relationships among all the taxa in the study.

In the initial analysis of this dataset, multi-dimensional scaling with the unweighted UniFrac distance \citep{lozupone2005unifrac} was used to visualize the samples.
Unweighted UniFrac is a phylogenetically-aware distance that only takes into account the presence or absence of a bacterial taxon.
In our reanalysis, we add weighted UniFrac (a phylogenetically-aware distance that takes into account abundance instead of presence/absence), two generalized Euclidean distances that use the phylogeny (the first based on double principal coordinates analysis (DPCoA) \citep{pavoine2004dissimilarities} and the second related to adaptive gPCA \citep{fukuyama2019adaptive}), and the standard Euclidean distance.

\begin{figure}
  \includegraphics[width=\textwidth]{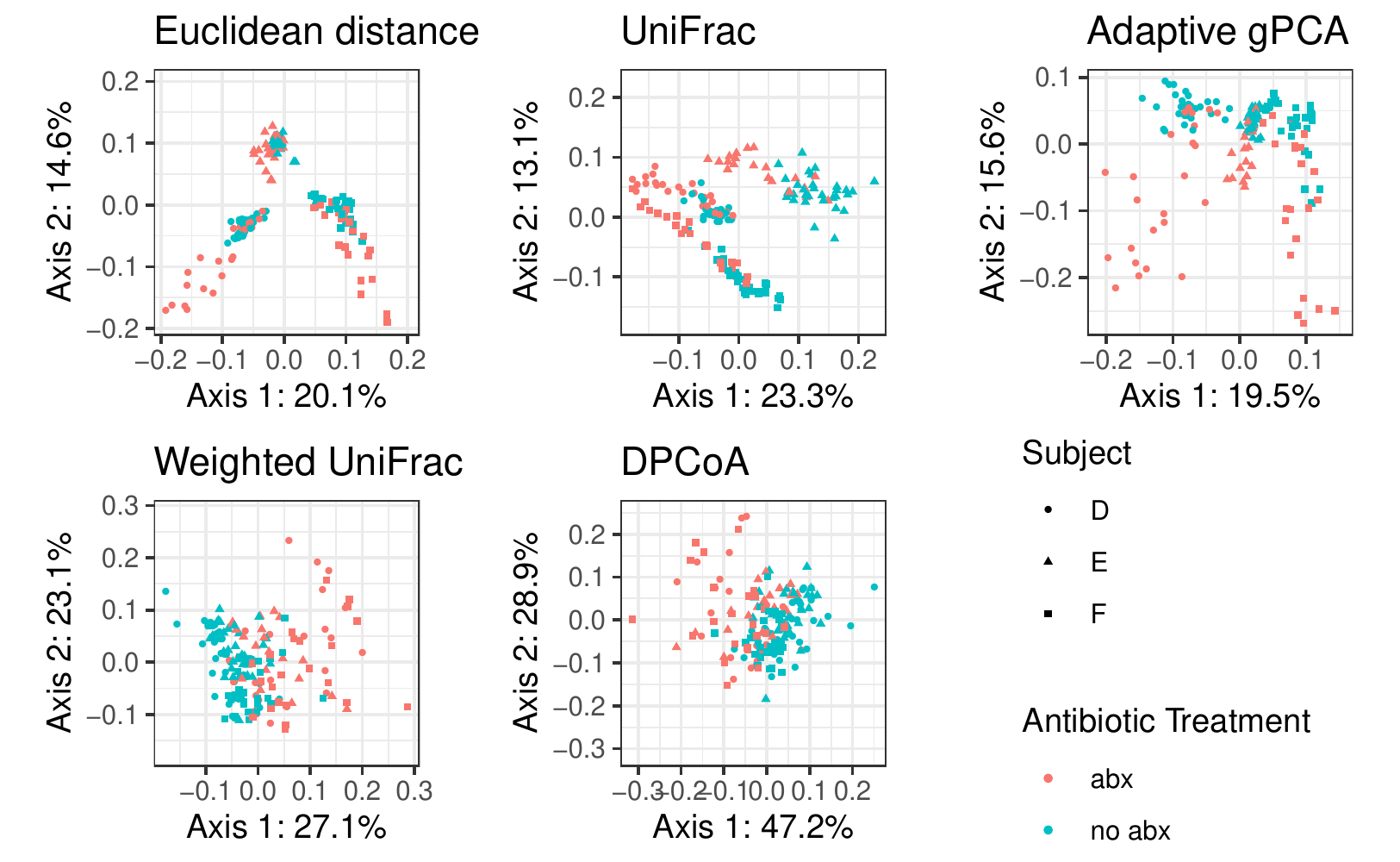}
  \caption{MDS plots of the microbiome data with five different distances. From top left: PCA or Euclidean distance, unweighted UniFrac, adaptive gPCA (a generalized Euclidean distance), weighted UniFrac, and DPCoA (a generalized Euclidean distance).
  The different distances give different representations of the samples, with different amounts of clustering by subject and by abx/no abx condition.}
  \label{Fig:all-samples-real}
\end{figure}

In Figure \ref{Fig:all-samples-real}, we see that each of the five distances gives a different representation of the samples.
Each panel in the figure shows the MDS embeddings of the samples with a different distance.
Each point represents one sample, with shape representing the subject and color related to whether the subject was on the antibiotic or not.
``Abx'' refers to samples taken  when the subject was taking the antibiotic and the week immediately after, while ``no abx'' refers to samples taken before the first course of the antibiotic, at least a week after the first course of the antibiotic but before the second course, and at least a week after the second course of the antibiotic.
With the Euclidean distance, there is a strong clustering of the samples by individual in the principal plane.
The abx/no abx conditions are offset from each other within the cluster for each individual, but the direction of the offset is different for the different subjects.
With unweighted UniFrac, there is still clustering of the samples by subjects, but the offset within each subject associated with the abx/no abx condition is now in a consistent direction for each subject, and is associated with the first MDS axis.
With the distance based on adaptive gPCA, we see a similar pattern as in unweighted UniFrac: clustering by subject, an offset within each subject associated with the abx/no abx condition in a consistent direction for each subject.
However, in adaptive gPCA, the first MDS axis is associated with subject, while the second is associated with the abx/no abx condition.
With weighted UniFrac and DPCoA we lose the distinct clusters associated with subject.
In each case there is an offset associated with the abx/no abx condition, but the magnitude of the offset is perhaps smaller than that seen with the other distances.

The differences between the representations provided by the five distances are intriguing: it seems that some distances favor an interpretation of the abx effect as being consistent from subject to subject while others favor an interpretation of a subject-specific effect; some distances favor an interpretation of the subjects falling into distinct clusters and others do not.
We can guess at the causes of these differences using the limited knowledge we have about the distances: The Euclidean distance does not use the phylogenetic information, while the others do.
It makes sense that the makeup of a subject's microbiota and the effect of an antibiotic would look more distinct when each taxon is viewed as equally distinct (implicit in the Euclidean distance) than it would when phylogenetic relationships are taken into account.
Looking at the local biplot axes for each MDS representation will allow us to sharpen this intuition.

\begin{figure}
  \includegraphics[width=\textwidth]{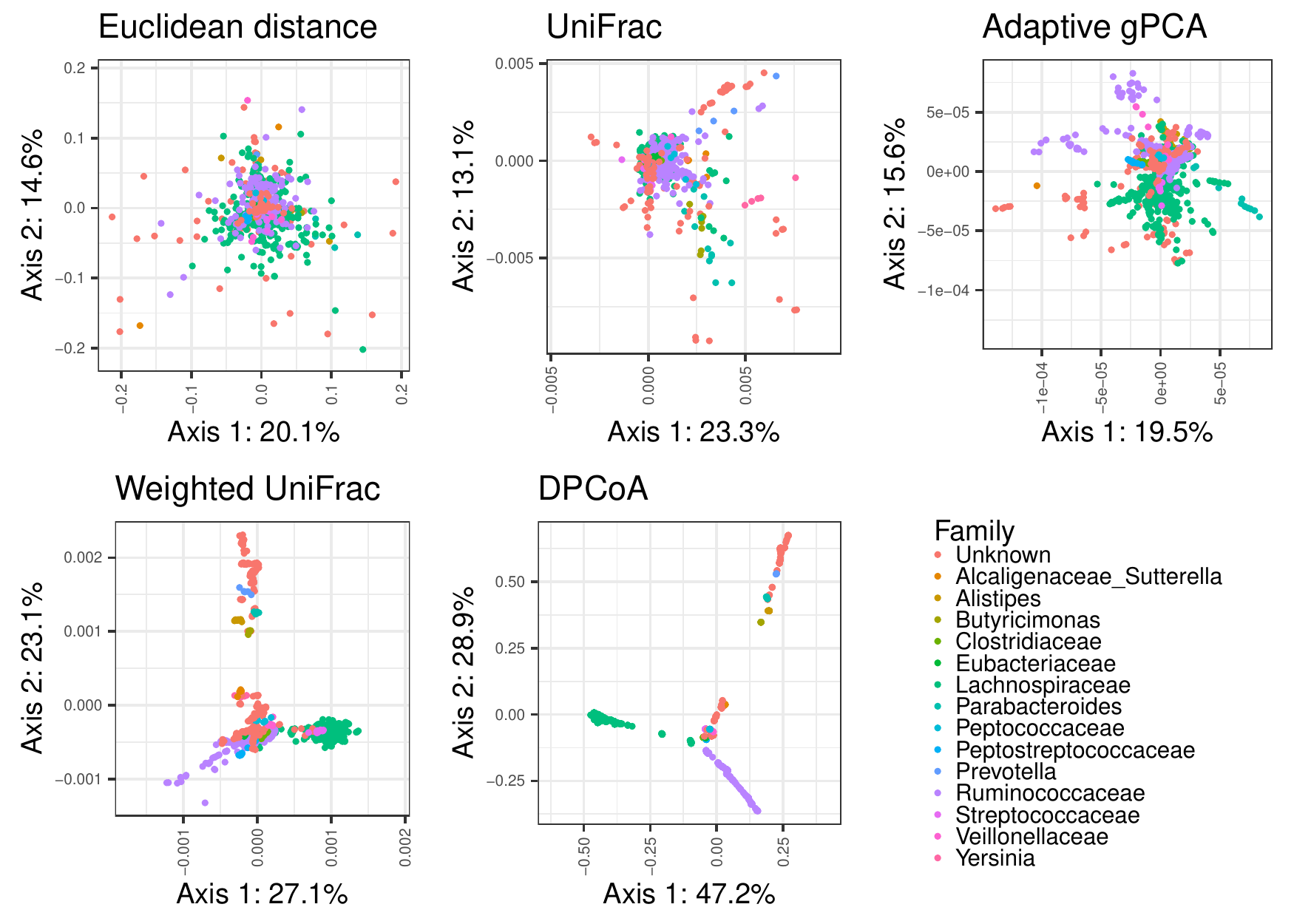}
  \caption{Local biplot axes at one point for MDS of the microbiome data with each of five different distances. From top left: LB axes for PCA/Euclidean distance (LB axes independent of location), UniFrac (LB axes for one of the data points embedding near the center of the plot), adaptive gPCA (a generalized Euclidean distance, so LB axes independent of location), weighted UniFrac (LB axes for one of the sample points embedding near the center of the plot), and DPCoA (a generalized Euclidean distance, so LB axes independent of location).
    Each point corresponds to the LB axis for one variable, which in this case are bacterial taxa.
    Points are colored according to the family the taxon is associated with.
  Note the difference in the relationship between taxonomic family and LB axis values: a strong association for DPCoA and weighted UniFrac, none for the Euclidean distance, and an intermediate association for the agPCA-based distance and unweighted UniFrac.}
  \label{Fig:all-biplots-real}
\end{figure}

The local biplot axes for MDS with the five distances, given in Figure \ref{Fig:all-biplots-real}, show a suggestive pattern.
Each panel corresponds to one of the MDS representations, and for each representation one set of local biplot axes is given.
Each point represents the local biplot axis for one variable, which in this case corresponds to a bacterial taxon.
The colors in the plot correspond to the family the bacterial taxon belongs to.
This information was not used in the construction of either the MDS plot or the local biplot axes, except (for the phylogenetically aware distances) to the extent that taxonomic family and phylogeny align.
The local biplot axes for the three generalized Euclidean distances (the first, third, and fifth panels) do not depend on the position in taxon space, and the axes shown are therefore representative of the whole space.
The local biplot axes for unweighted UniFrac and weighted UniFrac do depend on the position in taxon space.
For reasons of space, only one representative set of axes is shown for each, and in both cases the local biplot axes shown correspond to one of the samples that has an embedding near the origin in the embedding space.
An interactive {\tt shiny} \citep{shiny} app showing axes at arbitrary sample points is available at \url{https://jfukuyama.shinyapps.io/local-biplot-antibiotic-vis/}.

The first thing we notice about the local biplot axes for the different MDS representations are differences in the relationship between the local biplot axes and the taxonomy.
We see that for DPCoA and weighted UniFrac, the local biplot axes corresponding to taxa in the same family tend to have very similar values.
At the other extreme, with the Euclidean distance, there is no relationship between local biplot axis and family.
The local biplot axes for unweighted UniFrac and adaptive gPCA fall somewhere in the middle, with local biplot axes corresponding to taxa in the same family tending to have similar values, but not to the extent seen in weighted UniFrac or DPCoA.

Recalling our result in Theorem \ref{Thm:lb-centroid}, we can use the local biplot axes to interpret the MDS space for the Euclidean distance, adaptive gPCA, and DPCoA.
In each case, the location of a sample in the embedding space is a linear combination of the variables, with the weights given by the values of the local biplot axes.
Thus, we see that for the Euclidean distance, the sample embeddings are given by a linear combination of the variables where the weights are unrelated to the taxonomy.
On the other extreme, for DPCoA, the sample embeddings are given by a linear combination of the variables that can be approximately described as (on the first axis) a contrast between Lachnospiraceae and Ruminococcaceae/Unknown and (on the second axis) a contrast between Ruminococcaceae and unknown.
(Since the interesting offset that we see in the DPCoA sample embeddings between the abx/no abx conditions is a upper left vs. bottom right offset, perhaps that is actually the axis of interest and not the the first or second MDS axes.
Along the $y = -x$ line, the embedding of the samples can be described as a contrast between Lachnospiraceae and Ruminococcaceae.
These two families were in fact described in the initial paper as those showing a difference between the two conditions.)
Adaptive gPCA is somewhere in between the Euclidean distance and DPCoA: the sample embeddings can be well approximated as a linear combination that is approximately piecewise constant on small groups of closely related taxa.

Neither weighted UniFrac nor unweighted UniFrac is a generalized Euclidean distance, and the local biplot axes are not constant for either MDS representation.
However, the local biplot axes that we obtain suggest certain conclusions.
The local biplot axes shown for weighted UniFrac are approximately constant and approximately the same as those for DPCoA reflected over the $y$ axis.
This suggests that DPCoA and weighted UniFrac with multi-dimensional scaling are using the same internal representation of the samples, one related to the relative abundances of major divisions in the phylogenetic tree.
This is in line with the similarity of the embeddings given by MDS with weighted UniFrac and DPCoA: in each case, we see little separation of the samples but an offset between the abx/no abx conditions along the $y = x$ and $y = -x$ lines for weighted UniFrac and DPCoA, respectively.

Similarly, the qualitative results in the sample embeddings for MDS with unweighted UniFrac and adaptive gPCA (clustering of subjects and a consistent direction for the offset of the abx vs. no abx condition) mirrors the qualitative relationship between the local biplot axes and the taxonomic family for that pair (some relationship between taxonomic family and the value of the local biplot axis, but not as much as DPCoA/weighted UniFrac).

Overall, we see that the local biplot axes help us interpret the MDS embedding space.
In this particular instance, they suggest that the phylogenetic distances do some implicit smoothing of the data along the phylogenetic tree, and that the amount of smoothing increases as we go from Euclidean distance to unweighted UniFrac to adaptive gPCA to weighted UniFrac to DPCoA.
This is in line with previous work \citep{fukuyama2019emphasis} showing empiricially such a gradient in a larger set of phylogenetically-informed distances.
In addition to providing insight into general properties of the distances, the local biplot axes give us insight into the particular MDS representation at hand: which sets of bacterial taxa are likely to be over- or under-represented in particular samples or particular regions of the embedding space.
This is information that is not traditionally available in MDS but is of great importance to the consumers of such diagrams.


\section{Discussion}



We have introduced local biplot axes as a new tool for unboxing the black box that is multi-dimensional scaling.
We showed that these local biplot axes are a natural extension of PCA biplot axes and that there is a close connection between these axes and generalized principal components.
Our extension of the classic result on the equivalence of classical multi-dimensional scaling with the Euclidean distance and principal components gives us greater insight into the relationship between data space and embedding space for multi-dimensional scaling.
It particular, it suggests that for an arbitrary distance, if the local biplot axes are approximately constant, the distance can be well approximated by a generalized Euclidean distance and that the generalized Euclidean distance providing the approximation is a function of the local biplot axes.
Because generalized Euclidean distances can be understood as the standard Euclidean distance on a linear transformation of the input variables, this helps us interpret what variables or features in the data are important.

Even in the case where the local biplot axes are far from being constant on the input variable space, the local biplot axes allow us to investigate whether there are regions of local linearity (regions in which the local biplot axes are constant or approximately constant), and therefore where we might be able to approximate the distance as a generalized Euclidean distance or the MDS map as a linear map from data space to embedding space.
This gives us much more insight into the relationship between data and embedding space than is available either with MDS on its own (no information) or other proposals for biplots for MDS, which tend to be about a global linear approximation \citep{Satten2017-ds,wang-2019-gmd-biplot,greenacre2017ordination}.

Our simulated data example showed that local biplots can uncover substantive differences in the relationship between data space and embedding space, even in cases where MDS with two different distances gives the same embedding of the samples.
Our real data example showed that local biplots can suggest scientifically relevant reasons for the differences in the representations of the samples provided by MDS with different distances.
They can also suggest more general properties of the distances used: in this case, we saw that different phylogenetic distances appear to implicitly smooth the variables along the phylogenetic tree to different degrees, as has been suggested in other work \citep{fukuyama2019emphasis}.

There are several avenues for future work.
Perhaps most importantly, although our results suggest that approximately constant local biplot axes should imply that the distances are approximately of the form of a generalized Euclidean distance, we do not quantify the approximation or indeed how to choose the best generalized Euclidean distance for the approximation.
Similarly, to identify regions of local linearity, we would need a way to quantify the similarity of the local biplot axes in different regions of the space, and it is not immediately obvious what the best measure would be for this task.
Other potential areas for investigation include developing the relationship between distances and probabilistic models through the relationship between generalized Euclidean distances and the probabilistic interpretation of generalized PCA, using the LB axes as a starting point for confidence ellipses for MDS, and extending the ideas here to other types of low-dimensional embeddings.


\appendix

\section*{Appendix: Proofs}

\begin{proof}[Proof of Theorem \ref{Thm:pca-mds}]
This follows from Theorem \ref{Thm:quad-dist}, since the standard Euclidean distance $d(\mathbf x, \mathbf y)$ is the same as the generalized Euclidean distance $d_{\mathbf I}(\mathbf x, \mathbf y)$, and the singular value decomposition of $\mathbf X$ is the same as the generalized SVD of $(\mathbf X, \mathbf I, \mathbf I)$.
\end{proof}

\begin{proof}[Proof of Theorem \ref{Thm:gpca-existence}]
  This is a straightforward application of the spectral decomposition theorem.
  Let $\mathbf Q^{1/2}$ be a symmetric square root of $\mathbf Q$, and let $\mathbf Q^{-1/2}$ be its inverse.
  The spectral decomposition theorem tells us that we can find $\mathbf {\tilde V}$ and $\mathbf {\tilde \Lambda}$ such that
  \begin{align}
    \mathbf Q^{1/2} \mathbf X^T \mathbf D \mathbf X \mathbf Q^{1/2} = \mathbf {\tilde V} \mathbf {\tilde \Lambda} \mathbf {\tilde V}^T
  \end{align}
  with $\mathbf {\tilde V}^T \mathbf {\tilde V} = \mathbf I_p$, $\lambda_{11} \ge \lambda_{22} \ge \cdots \ge \lambda_{pp} \ge 0$.

  Then define  $\mathbf V := \mathbf Q^{-1/2} \mathbf {\tilde V}$ and $\mathbf \Lambda := \mathbf {\tilde \Lambda}$.
  We see that
  \begin{align}
    \mathbf X^T \mathbf D \mathbf X \mathbf Q \mathbf V &= \mathbf X^T \mathbf D \mathbf X \mathbf Q^{1/2} \mathbf {\tilde V} \\
    &= \mathbf Q^{-1/2} \mathbf {\tilde V} \mathbf{\tilde \Lambda} \mathbf {\tilde V}^T \mathbf {\tilde V} \\
    &= \mathbf V \mathbf {\Lambda}
  \end{align}
  and so $\mathbf V$ satisfies $\mathbf X^T \mathbf D \mathbf X \mathbf Q \mathbf V = \mathbf V \mathbf \Lambda$.

  Substitution also tells us that
  \begin{align}
    \mathbf V^T \mathbf Q \mathbf V &= \mathbf {\tilde V}^T \mathbf Q^{-1/2} \mathbf Q \mathbf Q^{-1/2} \mathbf {\tilde V} = \mathbf I_p
  \end{align}
  and so $\mathbf V$ and $\mathbf \Lambda$ have the required properties.
\end{proof}

\begin{lemma}
  Let $\mathbf X$ have centered columns, and let $d_{\mathbf Q}(\mathbf x, \mathbf y) = \sqrt{(\mathbf x - \mathbf y)^T \mathbf Q (\mathbf x - \mathbf y)}$.
  Let $\mathbf \Delta$, $\mathbf C_n$ be as defined in (\ref{Eq:squared-dist}).
  Then we have $-\frac{1}{2} \mathbf C_n \mathbf \Delta \mathbf C_n^T = \mathbf X \mathbf Q \mathbf X^T$.
  \label{Lem:JDJ}
\end{lemma}

\begin{proof}[Proof of Lemma \ref{Lem:JDJ}]
  Let $\mathbf x_i = \mathbf X_{i\cdot}^T$ be the column vector containing the $i$th row of $\mathbf X$.
  Let $\delta_{ij}$ denote $d_{\mathbf Q}(\mathbf x_i, \mathbf x_j)^2$, so that $\delta_{ij} = (\mathbf x_i - \mathbf x_j)^T \mathbf Q (\mathbf x_i - \mathbf x_j)= \mathbf x_i^T \mathbf Q \mathbf x_i + \mathbf x_j^T \mathbf Q \mathbf x_j - 2 \mathbf x_i^T \mathbf Q \mathbf x_j$.
  
  First note that
  \begin{align}
    \sum_{i=1}^n \delta_{ij} &= \sum_{i=1}^n \left( \mathbf x_i^T \mathbf Q \mathbf x_i + \mathbf x_j^T \mathbf Q \mathbf x_j - 2 \mathbf x_i^T \mathbf Q \mathbf x_j \right)\\
    &= n \mathbf x_j^T \mathbf Q \mathbf x_j + \sum_{i=1}^n( \mathbf x_i^T \mathbf Q \mathbf x_i ) -2 \mathbf x_j^T \mathbf Q \sum_{i=1}^n \mathbf x_i \\
    &= n \mathbf x_j^T \mathbf Q \mathbf x_j + \sum_{i=1}^n (\mathbf x_i^T \mathbf Q \mathbf x_i)
  \end{align}
  where the last line follows because we have taken $\mathbf X$ to have centered columns.

  We also have
  \begin{align}
    \sum_{j=1}^n \sum_{i=1}^n  \delta_{ij} &= \sum_{j=1}^n (n \mathbf x_j^T \mathbf Q \mathbf x_j + \sum_{i=1}^n \mathbf x_i^T \mathbf Q \mathbf x_i)  \\
    &= n \sum_{j=1}^n \mathbf x_j^T \mathbf Q \mathbf x_j + n \sum_{i=1}^n \mathbf x_i^T \mathbf Q \mathbf x_i
  \end{align}
  
  Then we have
  \begin{align}
    (-\frac{1}{2} \mathbf C_n \mathbf \Delta \mathbf C_n^T)_{ij} &= -\frac{1}{2} (\delta_{ij} - \frac{1}{n} \sum_{i=1}^n \delta_{ij} - \frac{1}{n} \sum_{j=1}^n \delta_{ij} + \frac{1}{n^2} \sum_{i,j} \delta_{ij} )\\
    &= -\frac{1}{2} (-2 \mathbf x_i^T \mathbf Q \mathbf x_j + \mathbf x_i^T \mathbf Q \mathbf x_i + \mathbf x_j^T \mathbf Q \mathbf x_j - \mathbf x_j^T \mathbf Q \mathbf x_j - \frac{1}{n} \sum_{i=1}^n (\mathbf x_i^T \mathbf Q \mathbf x_i) - \mathbf x_i^T \mathbf Q \mathbf x_i - \frac{1}{n}\sum_{j=1}^n (\mathbf x_j^T \mathbf Q \mathbf x_j) + \\
    &\quad \quad \frac{1}{n^2} (n \sum_{j=1}^n \mathbf x_j^T \mathbf Q \mathbf x_j + n \sum_{j=1}^n \mathbf x_i^T \mathbf Q \mathbf x_i) ) \\
    &= -\frac{1}{2} ( - 2 \mathbf x_i^T \mathbf Q \mathbf x_j)\\
    &= \mathbf x_i^T \mathbf Q \mathbf x_j
  \end{align}
  and so $-\frac{1}{2} \mathbf C_n \mathbf \Delta \mathbf C_n^T = \mathbf X \mathbf Q \mathbf X^T$.
\end{proof}

\begin{proof}[Proof of Theorem \ref{Thm:quad-dist}]
  By the definition,
  \begin{align}
LB(\mathbf z) &= \frac{1}{2} \begin{pmatrix}
    \frac{\partial }{\partial z_1} d_{\mathbf x_1}(\mathbf z)^2 & \cdots \frac{\partial}{\partial z_1} d_{\mathbf x_n}(\mathbf z)^2\\
    \vdots &  \vdots \\
    \frac{\partial}{\partial z_p} d_{\mathbf x_1}(\mathbf z)^2 & \cdots \frac{\partial}{\partial z_p} d_{\mathbf x_n}(\mathbf z)^2\\
\end{pmatrix} \mathbf M \mathbf \Lambda^{-1}\\
&= \frac{1}{2}\begin{pmatrix} \nabla d_{\mathbf x_1} (\mathbf z)^2 & \cdots & \nabla d_{\mathbf x_n} (\mathbf z)^2 \end{pmatrix} \mathbf M \mathbf \Lambda^{-1}
  \end{align}
  For generalized Euclidean distances, $d_{\mathbf x}(\mathbf z)^2 = d_{\mathbf Q}(\mathbf x, \mathbf z)^2 = (\mathbf x - \mathbf z)^T \mathbf Q (\mathbf x - \mathbf z)$.
  and so
  \begin{align}
    \nabla d_{\mathbf x}(\mathbf z)^2 = 2 \mathbf Q(\mathbf x - \mathbf z)
  \end{align}
  This gives us
  \begin{align}
    LB(\mathbf z) &= \frac{1}{2} \begin{pmatrix} 2 \mathbf Q (\mathbf x_1 - \mathbf z) & \cdots & 2 \mathbf Q (\mathbf x_n - \mathbf z) \end{pmatrix} \mathbf M \mathbf \Lambda^{-1}\\
    &= \mathbf Q (\mathbf X^T - \mathbf x \mathbf 1_n^T) \mathbf M \mathbf \Lambda^{-1} \\
    &= \mathbf Q \mathbf X^T \mathbf M \mathbf \Lambda^{-1} \\
    &= \mathbf Q \mathbf X^T \mathbf B \mathbf \Lambda^{-1/2}
  \end{align}
  where the second-to-last line follows because the matrix $\mathbf M$ has centered columns and the last line is a result of $\mathbf M = \mathbf B \mathbf \Lambda^{1/2}$.

  At this point, we see that the biplot axes have no dependence on the point $\mathbf x$.
  All that remains is to show the relationship to the gPCA axes.

  Let $\mathbf V = \mathbf X^T \mathbf B \mathbf \Lambda^{-1/2}$
  We have
  \begin{align}
    \mathbf X^T \mathbf X \mathbf Q \mathbf V &= \mathbf X^T \mathbf X \mathbf Q \mathbf X^T \mathbf B \mathbf \Lambda^{-1/2} \\
    &= \mathbf X^T (\mathbf B \mathbf \Lambda \mathbf B^T) \mathbf B \mathbf \Lambda^{-1/2} \\
    &= \mathbf X^T \mathbf B \mathbf \Lambda^{1/2} \\
    &= \mathbf V \mathbf \Lambda
  \end{align}
  and
  \begin{align}
    \mathbf V^T \mathbf Q \mathbf V &= (\mathbf X^T \mathbf B \mathbf \Lambda^{-1/2})^T \mathbf Q (\mathbf X^T \mathbf B \mathbf \Lambda^{-1/2})\\
    &= \mathbf \Lambda^{-1/2} \mathbf B^T \mathbf X \mathbf Q \mathbf X^T \mathbf B \mathbf \Lambda^{-1/2} \\
    &= \mathbf \Lambda^{-1/2} \mathbf B^T \mathbf B \mathbf \Lambda \mathbf \Lambda^{-1/2} \\
    &= \mathbf I
  \end{align}
  Therefore, $\mathbf V = \mathbf X^T \mathbf B \mathbf \Lambda^{-1/2}$ satisfies the conditions to be the normalized principal axes in gPCA of the triple $(\mathbf X, \mathbf Q, \mathbf I)$.
  If we substitute this equivalence in above, we see that
  \begin{align}
    LB(\mathbf z) &= \mathbf Q \mathbf V,
  \end{align}
  as desired
\end{proof}

\begin{proof}[Proof of Theorem \ref{Thm:constant-axes-on-input}]
  Let $f: \R^p \to \R^p$ be the map from data space to the full-dimensional embedding space.
  By our first assumption (that the MDS embedding space is of dimension $p$), $d_{\mathbf I_p}(f(\mathbf x_i), f(\mathbf x_j)) = d(\mathbf x_i, \mathbf x_j)$.
  By the second assumption (no additional dimensions required to embed $\mathbf z \in \R^p$), we have $d_{\mathbf I_p}(f(\mathbf x_i), f(\mathbf z)) = d(\mathbf x_i, \mathbf z)$.

  By definition of local biplot axes, the Jacobian of $f$ is $J_f(\mathbf z) = \mathbf L^T$, and if $f$ is the map from data space to embedding space defined in (\ref{Eq:mds-map}), we must have $f(\mathbf z) = \mathbf L^T \mathbf z + \mathbf k$ for some $\mathbf k \in \mathbf R^{p}$.
  Then for any $\mathbf x_i$, $i = 1,\ldots, n$, $\mathbf z \in \R^p$, we have
  \begin{align}
    d(\mathbf x_i, \mathbf z) &= d_{\mathbf I_p}(f(\mathbf x_i), f(\mathbf z)) \\
    &= d_{\mathbf I_p}(\mathbf L^T \mathbf x_i + \mathbf k, \mathbf L^T \mathbf z + \mathbf k) \\
    &= \sqrt{(\mathbf x_i - \mathbf z)^T \mathbf L \mathbf L^T (\mathbf x_i - \mathbf z) }\\
    &= d_{\mathbf L \mathbf L^T}(\mathbf x_i, \mathbf z),
  \end{align}
  as desired.  
\end{proof}

\begin{proof}[Proof of Theorem \ref{Thm:constant-axes-with-dist-assumptions}]
  Proof is by induction.
  We will show that for any $k = 1,\ldots, n$, if
  \begin{align}
    d(\mathbf z_1, \mathbf z_2)= d_{\mathbf L \mathbf L^T}(\mathbf z_1, \mathbf z_2) \text{ if } \mathbf z_1, \mathbf z_2 \in \text{span}(\mathbf x_1,\ldots, \mathbf x_k)\label{Eq:ind-k}
  \end{align}
  
  Base case: We show \eqref{Eq:ind-k} for $k = 1$.
  If $\mathbf z_i \in \text{span}(\mathbf x_1)$, then $\mathbf z_i = \alpha_i \mathbf x_1$, $\alpha_i \in \R$.
  \begin{align}
    d(\mathbf z_1, \mathbf z_2) &= d(\alpha_1 \mathbf x_1, \alpha_2 \mathbf x_1) \\
    &= (\alpha_1 - \alpha_2) d(\mathbf x_1, \mathbf 0_p) \\
    &= (\alpha_1 - \alpha_2)d_{\mathbf L \mathbf L^T}(\mathbf x_1, \mathbf 0_p) \\
    &= d_{\mathbf L \mathbf L^T} ((\alpha_1 - \alpha_2) \mathbf x_1, \mathbf 0_p)\\
    &= d_{\mathbf L \mathbf L^T} (\mathbf z_1, \mathbf z_2),
  \end{align}
  where the second line follows by translation invariance and homogeneity of $d$, the third line follows from Theorem \ref{Thm:constant-axes-on-input}, and the remainder is translation invariance, homogeneity, and the definition of $\mathbf z_i$.

  Induction step: Suppose that \eqref{Eq:ind-k} holds for some $k \in \{1,\ldots, n-1\}$, and let $\mathbf z_i, \mathbf z_j \in \text{span}(\mathbf x_1,\ldots, \mathbf x_{k+1})$.
  We can write $\mathbf z_i = \mathbf z_i^k + \alpha_i \mathbf x_{k+1}$ for $\alpha_i \in \R$, $\mathbf z_i^k \in \text{span}(\mathbf x_1, \ldots, \mathbf x_k)$, $i = 1,2$.
  If $\alpha_1 = \alpha_2 = 0$, then $\mathbf z_i, \mathbf z_j \in \text{span}(\mathbf x_1, \ldots, \mathbf x_k)$, and the result holds by the assumption \eqref{Eq:ind-k}.
  If either $\alpha_1$ or $\alpha_2 \ne 0$, then
  \begin{align}
    d(\mathbf z_1, \mathbf z_2) &= d(\mathbf z_1^k + \alpha_1 \mathbf x_{k+1}, \mathbf z_2^k + \alpha_2 \mathbf x_{k+1}) \\
    &= \begin{cases}
      (\alpha_2 - \alpha_1) d((\alpha_2 - \alpha_1)^{-1} (\mathbf z_1^k - \mathbf z_2^k), \mathbf x_{k+1}) & \alpha_2 -\alpha_1 \ne 0 \\
      d(\mathbf z_1^k, \mathbf z_2^k) & \alpha_1 = \alpha_2
      \end{cases}
  \end{align}
  In the case $\alpha_1 = \alpha_2$, we have
  \begin{align}
    d(\mathbf z_1^k, \mathbf z_2^k) &= d_{\mathbf L \mathbf L^T}(\mathbf z_1^k, \mathbf z_2^k) \\
    &= d_{\mathbf L \mathbf L^T}(\mathbf z_1^k + \alpha_1 \mathbf x_{k+1}, \mathbf z_1^k + \alpha_2 \mathbf x_{k+1}) \\
    &= d_{\mathbf L \mathbf L^T}(\mathbf z_1, \mathbf z_2),
  \end{align}
  where the first line follows by \eqref{Eq:ind-k}, the second by translation invariance, and the third by the definition of $\mathbf z_i$.

  Otherwise, we have
  \begin{align}
    (\alpha_2 - \alpha_1) d((\alpha_2 - \alpha_1)^{-1} (\mathbf z_1^k - \mathbf z_2^k), \mathbf x_{k+1}) &=     (\alpha_2 - \alpha_1) d_{\mathbf L \mathbf L^T}((\alpha_2 - \alpha_1)^{-1} (\mathbf z_1^k - \mathbf z_2^k), \mathbf x_{k+1}) \\
    &= d_{\mathbf L \mathbf L^T}(\mathbf z_1, \mathbf z_2)
  \end{align}
  where the first line follows from Theorem \ref{Thm:constant-axes-on-input} and the second line is algebra.
  Thus \eqref{Eq:ind-k} holds for $k$, (\ref{Eq:ind-k}) holds for $k + 1$.

  Finally, since $\text{span}(\mathbf x_1,\ldots, \mathbf x_n) = \R^p$ by assumption, the case $k = n$ implies that for any $\mathbf z_1, \mathbf z_2 \in \R^p$, $d(\mathbf z_1, \mathbf z_2) = d_{\mathbf L \mathbf L^T}(\mathbf z_1, \mathbf z_2)$ as desired.

\end{proof}

\begin{proof}[Proof of Theorem \ref{Thm:gower-centroid}]
  Let $\mathbf z = \sum_{j=1}^p \alpha_j \mathbf e_j$ be the supplemental point.
  We can rewrite the distances between $\mathbf x_i$ and $\alpha_j \mathbf e_j$ in terms of their components:
  \begin{align}
    d(\mathbf x_i, \alpha_j \mathbf e_j)^2 &=   h(x_{ij}, \alpha_j) - h(x_{ij}, 0) + \sum_{k=1}^p h(x_{ik}, 0) \\
    \sum_{j=1}^p d(\mathbf x_i, \alpha_j \mathbf e_j)^2 &= \sum_{j=1}^p h(x_{ij}, \alpha_j) - \sum_{j=1}^p h(x_{ij}, 0) + p \sum_{k=1}^p h(x_{ik}, 0)\\
    &= \sum_{j=1}^p h(x_{ij}, \alpha_j) + (p - 1) \sum_{k=1}^p h(x_{ik}, 0)
  \end{align}
  This allows us to rewrite the distance between the samples and the supplemental points in terms of the distances between $\mathbf x_i$ and $\alpha_j \mathbf e_j$.
  \begin{align}
    d(\mathbf x_i, \mathbf z)^2 &= \sum_{j=1}^p h(x_{ij}, \alpha_j) \\
    &= \sum_{j=1}^p d(\mathbf x_i, \alpha_j \mathbf e_j)^2 - (p-1)\sum_{k=1}^p h(x_{ik}, 0) \\
    &= \sum_{j=1}^p d(\mathbf x_i, \alpha_j \mathbf e_j)^2 - (p-1) d(\mathbf x_i, \mathbf 0_p)^2
  \end{align}
  Finally, let $\mathbf a \in \R^n$ be defined as in (\ref{Eq:a-def}): $a_i =  (-\frac{1}{2} \mathbf C_n \mathbf \Delta \mathbf C_n^T)_{ii} - d(\mathbf x_i, \mathbf z)^2$.
  Then we can rewrite $\mathbf a$ in terms of distances to $\alpha_j \mathbf e_j$ and distances to $\mathbf 0_p$:
  \begin{align}
    a_i &= (-\frac{1}{2} \mathbf C_n \mathbf \Delta \mathbf C_n^T)_{ii} - d(\mathbf x_i, \mathbf z)^2 \\    
    &= (-\frac{1}{2} \mathbf C_n \mathbf \Delta \mathbf C_n^T)_{ii} - \sum_{j=1}^p d(\mathbf x_i, \alpha_j \mathbf e_j)^2 - (p-1)d(\mathbf x_i, \mathbf 0_p)^2\\
    &= \left[\sum_{j=1}^p ( (-\frac{1}{2} \mathbf C_n \mathbf \Delta \mathbf C_n^T)_{ii} - d(\mathbf x_i, \alpha_j \mathbf e_j)^2)\right] - (p-1)\left[ (-\frac{1}{2} \mathbf C_n \mathbf \Delta \mathbf C_n^T)_{ii} - d(\mathbf x_i, \mathbf 0_p)^2\right]
  \end{align}
  Which finally allows us to write the map of a supplemental point to the embedding space in terms of the embeddings of $\alpha_j \mathbf e_j$ and the embedding $\mathbf 0_p$:
  \begin{align}
    f(\mathbf z) &= \frac{1}{2} \mathbf \Lambda_{1:k,1:k}^{-1/2} \mathbf M^T \mathbf a \\
    &= \frac{1}{2} \mathbf \Lambda_{1:k,1:k}^{-1/2} \mathbf M^T \left[\sum_{j=1}^p ( (-\frac{1}{2} \mathbf C_n \mathbf \Delta \mathbf C_n^T)_{ii} - d(\mathbf x_i, \alpha_j \mathbf e_j)^2)\right] - \\
    &\quad\quad \frac{p-1}{2} \mathbf \Lambda_{1:k,1:k}^{-1/2} \mathbf M^T \left[ (-\frac{1}{2} \mathbf C_n \mathbf \Delta \mathbf C_n^T)_{ii} - d(\mathbf x_i, \mathbf 0_p)^2\right]\\
    &= \sum_{j=1}^p \frac{1}{2} \mathbf \Lambda_{1:k,1:k}^{-1/2} \mathbf M^T \left[ (-\frac{1}{2} \mathbf C_n \mathbf \Delta \mathbf C_n^T)_{ii} - d(\mathbf x_i, \alpha_j \mathbf e_j)^2)\right] - (p-1) f(\mathbf 0_p)\\
    &= \sum_{j=1}^p f(\alpha_j \mathbf e_j) - (p-1)f(\mathbf 0_p)
  \end{align}
  Since $\mathbf c = \frac{1}{p} \sum_{j=1}^p f(\alpha_j \mathbf e_j)$, this implies that $f(\mathbf z) = p \mathbf c - (p-1)f(\mathbf 0_p)$, as desired.
\end{proof}

\begin{proof}[Proof of Theorem \ref{Thm:lb-centroid}]
  The result follows if we can show that $f(\mathbf z) = \mathbf V^T \mathbf Q \mathbf z$, because in that case we will have
  \begin{align}
    f(\mathbf z) = \mathbf V^T \mathbf Q \mathbf z =  \mathbf V^T \mathbf Q \sum_{j=1}^p \alpha_j \mathbf e_j = \sum_{j=1}^p \alpha_j (\mathbf Q \mathbf V)_{\cdot, j}^T
  \end{align}

  Since $J_f = \mathbf Q \mathbf V$ (by Theorem \ref{Thm:quad-dist} and the definition of local biplot axes), we know that $f(\mathbf z) = \mathbf V^T \mathbf Q \mathbf z + \mathbf k$, and we just need to show that $\mathbf k = \mathbf 0_k$.

  Note that
  \begin{align}
    \sum_{i=1}^n f(\mathbf x_i) &= \sum_{i=1}^n \mathbf V^T \mathbf Q \mathbf x_i + n \mathbf k = \mathbf V^T \mathbf Q \mathbf X^T  \mathbf 1_n + n \mathbf k = n \mathbf k \label{Eq:k}
  \end{align}
  because $\mathbf X$ has centered columns.
  We also have
  \begin{align}
    \sum_{i=1}^n f(\mathbf x_i) &=  \begin{pmatrix} f(\mathbf x_1) & \cdots & f(\mathbf x_n) \end{pmatrix}\mathbf 1_n \\
    &= \mathbf \Lambda^{1/2}_{1:k, 1:k} (\mathbf B_{\cdot, 1:k})^T \mathbf 1_n
  \end{align}
  because the embeddings $f(\mathbf x_1), \ldots, f(\mathbf x_n)$ in the first $k$ dimensions are given in the rows of $\mathbf B_{1:k} \mathbf \Lambda_{1:k,1:k}^{1/2}$.
  Then since $\|\mathbf \Lambda^{1/2} \mathbf B^T \mathbf 1_n \|_2^2 = \mathbf 1^T \mathbf B \mathbf \Lambda \mathbf B^T \mathbf 1_n= \mathbf 1^T (-\frac{1}{2} \mathbf C_n \mathbf \Delta \mathbf C_n^T)\mathbf 1 = 0$ (Equation \ref{Eq:JDJ} and the definition of $\mathbf C_n$), $ \mathbf \Lambda\mathbf B^T \mathbf 1_n = \mathbf 0$, and so $\sum_{i=1}^n f(\mathbf x_i) = \mathbf 0_k$. Combined with (\ref{Eq:k}), we have $\mathbf k = \mathbf 0_k$, as desired.
\end{proof}

\bibliography{local-biplots}

\end{document}